\documentclass[10pt]{article}

\usepackage{kerger}

\begin{document}

%different title options

%\title{Algorithms and complexity analysis in the Quantum and Classical CONGEST-CLIQUE Models for Approximate Steiner Trees and Directed Minimum Spanning Trees}

%\title{Quantum-Accelerated Distributed Algorithms for Approximate Steiner Trees and Directed Minimum Spanning Trees}

\title{Mind the $\tilde{\mathcal{O}}$: asymptotically better, but still impractical, quantum distributed algorithms}

\author[1,2,3]{Phillip A. Kerger}

\author[2,3,4]{David E. Bernal Neira}

\author[2,3]{Zoe Gonzalez Izquierdo}

\author[3]{Eleanor G. Rieffel}

\affil[1]{Department of Applied Mathematics and Statistics, Johns Hopkins University}

\affil[2]{Research Institute of Advanced Computer Science, USRA}

\affil[3]{Quantum Artificial Intelligence Laboratory, NASA Ames Research Center}

\affil[4]{Davidson School of Chemical Engineering, Purdue University}

\maketitle

\begin{abstract}

The CONGEST and CONGEST-CLIQUE models have been carefully studied to represent situations where the communication bandwidth between processors in a network is severely limited. Messages of only $\mathcal{O}(\log(n))$ bits of information each may be sent between processors in each round. The quantum versions of these models allow the processors instead to communicate and compute with quantum bits under the same bandwidth limitations. This leads to the following natural research question: What problems can be solved more efficiently in these quantum models than in the classical ones? Building on existing work, we contribute to this question in two ways. Firstly, we present two algorithms in the Quantum CONGEST-CLIQUE model of distributed computation that succeed with high probability; one for producing an approximately optimal Steiner Tree, and one for producing an exact directed minimum spanning tree, each of which uses $\tilde{\mathcal{O}}(n^{1/4})$ rounds of communication and $\tilde{\mathcal{O}}(n^{9/4})$ messages, where $n$ is the number of nodes in the network. The algorithms thus achieve a lower asymptotic round and message complexity than any known algorithms in the classical CONGEST-CLIQUE model. At a high level, we achieve these results by combining classical algorithmic frameworks with quantum subroutines. An existing framework for using a distributed version of Grover's search algorithm to accelerate triangle finding lies at the core of the asymptotic speedup. Secondly, we carefully characterize the constants and logarithmic factors involved in our algorithms as well as related algorithms, otherwise commonly obscured by $\tilde{O}$ notation. The analysis shows that some improvements are needed to render both our and existing related quantum and classical algorithms practical, as their asymptotic speedups only help for very large values of $n$.
  %old abstract:
    %We present two algorithms in the Quantum CONGEST-CLIQUE model of distributed computation that succeed with high probability; one for producing an approximately optimal Steiner Tree, and one for producing an exact directed minimum spanning tree, each of which uses $\tilde{O}(n^{1/4})$ rounds of communication 
    %and $\tilde{O}(n^{9/4})$ messages, achieving a lower asymptotic round and message complexity than any known algorithms in the classical CONGEST-CLIQUE model. At a high level, we achieve these results by combining classical algorithmic frameworks with quantum subroutines. Additionally, we characterize the constants and logarithmic factors involved in our algorithms as well as related classical algorithms,  otherwise obscured by $\tilde{O}$ notation, revealing that advances are needed to render both the quantum and classical algorithms practical.

\end{abstract}

\section{Introduction}
The classical CONGEST-CLIQUE Model (cCCM henceforth) in distributed computing has been carefully studied as a model central to the field, e.g., \cite{TowardsCCM_Complexity_KorhonenSumoela2018, Saikia_Karmakar2019_SteinerTree_CONGESTCLIQUE, fischer2021_DMST,Lenzen2012_OptimalRoutingSortingInCongestClique, Dolev_lenzen_2012TriTA, MST_in_O1_CCM_Nowicki2019}. 
In this model, processors 
in a network solve a problem whose input is distributed across the nodes under significant communication limitations, described in detail in \S \ref{sec: background}. 
For example, a network of aircraft or spacecraft, satellites, and control stations, all with large distances between them, may have severely limited communication bandwidth to be modeled in such a way. 
The quantum version of this model, in which quantum bits can be sent between processors, the quantum CONGEST-CLIQUE Model (qCCM), as well as the quantum CONGEST model, have been the subject of recent research \cite{IzumiLeGallMag2019_APSP_QuantumDist, CensorHillel2022_QuantDistCliqueDetect, vanApeldoorn2022_DistQuantQueriesCONGEST, Elkin2014_NoQuantumSpeedups} in an effort to understand how quantum communication may help in these distributed computing frameworks. 
For the quantum CONGEST Model, however, \cite{Elkin2014_NoQuantumSpeedups} showed that many problems cannot be solved more quickly than in the classical model.
These include shortest paths, minimum spanning trees, Steiner trees, min-cut, and more; 
the computational advantages of quantum communication are thus severely limited in the CONGEST setting, though a notable positive result is sub-linear diameter computation in \cite{LeGall_Mag2018_Diameter_DistGrover}.
No comparable negative results exist for the qCCM, and in fact, \cite{IzumiLeGallMag2019_APSP_QuantumDist} provides an asymptotic quantum speedup for computing all-pairs shortest path (APSP henceforth) distances.
%, despite the negative result in the CONGEST (non-clique) setting.
Hence, it is apparent that the negative results of \cite{Elkin2014_NoQuantumSpeedups} cannot transfer over to the qCCM, so investigating these problems in the qCCM presents an opportunity for contribution to the understanding of how quantum communication may help in these distributed computing frameworks.
In this paper, we contribute to this understanding by formulating algorithms in the qCCM for finding approximately optimal Steiner trees and exact directed minimum spanning trees using $\tilde{\mathcal{O}}(n^{1/4})$ rounds -- asymptotically fewer rounds than any known classical algorithms. 
This is done by augmenting the APSP algorithm of \cite{IzumiLeGallMag2019_APSP_QuantumDist} with an efficient routing table scheme, which is necessary to make use of the shortest {\em paths} information instead of only the APSP {\em distances}, and using the resulting subroutine with existing classical algorithmic frameworks.
Beyond asymptotics, we also characterize the complexity of our algorithms as well as those of \cite{IzumiLeGallMag2019_APSP_QuantumDist, CensorHillel2016_AlgebraicMethodsCCM_fast_APSP, Saikia_Karmakar2019_SteinerTree_CONGESTCLIQUE, fischer2021_DMST} to include the logarithmic and constant factors involved to estimate the scales at which they would be practical, which was not included in the previous work. 
It should be noted that, like APSP, these problems cannot see quantum speedups in the CONGEST (non-clique) setting as shown in \cite{Elkin2014_NoQuantumSpeedups}.
Our Steiner tree algorithm is approximate and based on a classical polynomial-time centralized algorithm of \cite{Kou1981_fast_algo_for_ST_}.
Our directed minimum spanning tree problem algorithm follows an approach similar to \cite{fischer2021_DMST}, which effectively has its centralized roots in \cite{lovasz1985_DMST}.

\section{Background and Setting}\label{sec: background}
This section provides the necessary background for our algorithms' settings and the problems they solve. 

\subsection{The CONGEST and CONGEST-CLIQUE Models of Distributed Computing}
In the standard CONGEST model, we consider a graph of $n$ processor nodes whose edges represent communication channels. Initially, each node knows only its neighbors in the graph and associated edge weights.
%, but no node has global information about the graph.
In rounds, each processor node executes computation locally and then communicates with its neighbors before executing further local computation.
The congestion limitation restricts this communication, with each node able to send only one message of $\mathcal{O}(\log(n))$ classical bits in each round to its neighbors, though the messages to each neighbor may differ.
Since there are $n$ nodes, assigning them ID labels $1, ..., n$ means the binary encoding size of a label is $\ceil{\log(n)}$ bits -- i.e., each message in CONGEST of $O(\log n)$ bits can contain roughly the amount of information to represent one node ID.
In the cCCM, we separate the communication graph from the problem input graph by allowing all nodes to communicate with each other, though the same $\mathcal{O}(\log(n))$ bits-per-message congestion limitation remains.
%Since all nodes can communicate with one another each round and messages to different nodes can be different, i.e., 
Hence, a processor node could send $n-1$ different messages to the other $n-1$ nodes in the graph, with a single node distributing up to $\mathcal{O}(n\cdot\log(n))$ bits of information in a single round.
Taking advantage of this way of dispersing information to the network is paramount in many efficient CONGEST-CLIQUE algorithms.
%The diameter of the network is also no longer a hindrance to information communication, so the cCCM is significantly more powerful than the standard CONGEST model for many tasks. 
The efficiency of algorithms in these distributed models is commonly measured in terms of the \textit{round complexity}, the number of rounds of communication used in an algorithm to solve the problem in question. 
% Under this complexity measure, many algorithms in the cCCM outperform their counterparts in the standard CONGEST model. 
A good overview of these distributed models can be found in \cite{DistGraphAlgos_LectureNotes_Ghaffari}. 
%, and in \S \ref{sec: background}, we will also provide their precise definitions along with their quantum counterparts. 

%In the quantum models of CONGEST and cCCM, the processors can send messages of $O(\log(n))$ \textit{quantum bits} to one another in place of classical bits. 

%A recent body of research investigates questions surrounding what problems can be solved more efficiently, i.e., using fewer rounds, in these quantum models (\cite{IzumiLeGall2020_TriangleFinding}, \cite{LeGall_Mag2018_Diameter_DistGrover}, \cite{IzumiLeGallMag2019_APSP_QuantumDist}, \cite{CensorHillel2022_QuantDistCliqueDetect}, \cite{Elkin2014_NoQuantumSpeedups}, \cite{vanApeldoorn2022_DistQuantQueriesCONGEST}). 

%This paper contributes to understanding the differences between the cCCM and qCCM by providing algorithms for the Steiner Tree and the Directed Minimum Spanning Tree problems in the qCCM that improve upon the asymptotic complexities of any known classical algorithms.
%in the classical setting. 

\subsection{Quantum Versions of CONGEST and CONGEST-CLIQUE}\label{sec: intro - quantum versions of DC}
The quantum models we work in are obtained via the following modification:
Instead of restricting to messages of $\mathcal{O}(\log(n))$ classical bits, we allow messages to consist of $\mathcal{O}(\log(n))$ quantum bits, qubits. For background on qubits and the fundamentals of quantum computing, we refer the reader to \cite{RieffelGentleIntroductionToQuantum}. We formally define the qCCM, the setting for our algorithms, as follows: 
\begin{definition}[Quantum CONGEST-CLIQUE] \label{def: QCCM}
The Quantum CONGEST-CLIQUE Model (qCCM) is a distributed computation model in which an input graph $G = (V, E, W)$ is distributed over a network of $n$ processors, where each processor is represented by a node in $V$.
Each node is assigned a unique ID number in $[n]$. Time passes in \textit{rounds}, each of which consists of the following:
\begin{enumerate}
    \item Each node may execute unlimited local computation.
    \item Each node may send a message consisting of either a register of $\mathcal{O}(\log n)$ qubits or a string of $\mathcal{O}(\log n)$ classical bits to each other node in the network. Each of those messages may be distinct. 
    \item Each node receives and saves the messages the other nodes send it. 
\end{enumerate}
The input graph $G$ is distributed across the nodes as follows: 
Each node knows its own ID number, the ID numbers of its neighbors in $G$, the number of nodes $n$ in $G$, and the weights corresponding to the edges it is incident upon. The output solution to a problem must be given by having each node $v\in V$ return the restriction of the global output to $\mathcal{N}_{G}(u) := \{v: uv \in E\}$, its neighborhood in $G$. No entanglement is shared across nodes initially.

\end{definition}

This is an analog of the cCCM, except that quantum bits may be sent in place of classical bits.
To clarify the output requirement, in the Steiner tree problem, we require node $u$ to output the edges of the solution tree that are incident upon $u$. 
Since many messages in our algorithms need not be sent as qubits, we define the qCCM slightly unconventionally, allowing either quantum or classical bits to be sent.  We specify those that may be sent classically. 
However, even without this modification, the quantum versions of CONGEST and cCCM are at least as powerful as their classical counterparts. This is because any $n$-bit classical message can be instead sent as an $n$-qubit message of unentangled qubits; 
for a classical bit reading $0$ or $1$, we can send a qubit in the state $\ket{0}$ or $\ket{1}$ respectively, and then take measurements with respect to the $\{\ket{0}, \ket{1}\}$ basis to read the same message the classical bits would have communicated. 
Hence, one can also freely make use of existing classical algorithms in the qCCM.
Further, the assumption that IDs are in $[n]$, with $n$ known, is not necessary but is convenient; without this assumption, we could have all  nodes broadcast their IDs to the entire network and then assign a new label in $[n]$ to each node according to an ordering of the original IDs, resulting in our assumed situation. 
% Regarding how the output is returned, for many problems, including any problems whose output is a tree, the entire output may equivalently be reported at a single predetermined node $r\in V$. 
% Since the solution tree consists of at most $n-1$ edges, there are $2(n-1)$ messages of length $\log(n)$ that need to be sent to node $r$. 
% Hence, this can be achieved by Lenzen's routing scheme \cite{Lenzen2012_OptimalRoutingSortingInCongestClique} within four rounds. 

\begin{remark}\label{remark: information storage}
Definition \ref{def: QCCM} does not account for how the information needs to be stored. In this paper, it suffices for all information regarding the input graph to be stored classically as long as there is quantum access to that data. We provide some details on this in \S \ref{app:information access} of the appendix.
%In particular, it needs to be possible for nodes to evaluate the relevant instances of \eqref{eq: inequality for triangle search} quantumly, for which it suffices to have the relevant edge weights stored classically as long as they can be quantumly accessed. 
\end{remark}

\begin{remark}
No entanglement being shared across nodes initially in definition \ref{def: QCCM} results in quantum teleportation not being a trivial way to solve problems in the qCCM.
\end{remark}

\example \label{ex: grover example}
To provide some intuition on how allowing communication through qubits in this distributed setting can be helpful, we now describe and give an example of distributed Grover search, first described in \cite{LeGall_Mag2018_Diameter_DistGrover}.
The high-level intuition for why quantum computing gives an advantage for search is that quantum operations use quantum interference effects to have canceling effects among non-solutions.
Grover search has a generalization called ``amplitude amplification'' we will use; see \cite{RieffelGentleIntroductionToQuantum} for details on these algorithms.
Now, for a processor node $u$ in the network and a Boolean function $g:X\rightarrow\{0,1\}$, suppose there exists a classical procedure $\mathcal{C}$ in the cCCM that allows $u$ to compute $g(x)$, for any $x\in X$ in $r$ rounds.
The quantum speedup will come from computing $\mathcal{C}$ in a quantum superposition, which enables $g$ to be evaluated with inputs in superposition so that amplitude amplification can be used for inputs to $g$. Let $A_i: \{x\in X:g(x)=i\}$, for $i=0, 1$, and suppose that $0<|A_1|\leq|X|/2$. Then classically, node $u$ can find an $x\in A_1$ in $\Theta(r|X|)$ rounds by checking each element of $X$.
Using the quantum distributed Grover search of \cite{LeGall_Mag2018_Diameter_DistGrover} enables $u$ to find such an $x$ with high probability in only $\tilde{\mathcal{O}}(r\sqrt{|X|})$ rounds by evaluating the result of computing $g$ on a superposition of inputs. 

We illustrate this procedure in an example case where a node $u$ wants to inquire whether one of its edges $uv$ is part of a triangle in $G$. We first describe a classical procedure for this, followed by the corresponding quantum-distributed search version. For $v\in \mathcal{N}_G(u)$, denote by $\mathcal{I}_v:V\rightarrow \{0,1\}$ the indicator function of $\mathcal{N}_{G}(v)$, and by $g_{uv}:\mathcal{N}_G(u)\rightarrow\{0,1\}$ its restriction to inputs in $\mathcal{N}_G(u)$. 
Classically, node $u$ can evaluate $g_{uv}(w)$ in two rounds for any $w\in\mathcal{N}_G(u)$ by sending the ID of $w$ (of length $\log n$) to $v$, and having $v$ send back the answer $\mathcal{I}_v(w)$. Then $u$ can check $g_{uv}(w)$ for each $w\in\mathcal{N}_G(u)$ one at a time to determine whether $uv$ is part of a triangle in $G$ or not in $2\cdot|\mathcal{N}_G(u)|$ rounds. 
 
For the distributed quantum implementation, $u$ can instead initialize a register of $\log n$ qubits as $|\psi\rangle_0:= \frac{1}{\sqrt{|\mathcal{N}_G(u)|}}\sum_{x\in\mathcal{N}_G(u)}|x\rangle$, all the inputs for $g_{uv}$ in equal superposition.
To do a Grover search, $u$ needs to be able to evaluate $g_{uv}$ with inputs $|\psi\rangle$ in superposition. For the quantum implementation of $\mathcal{C}$, $u$ sends a quantum register in state $|\psi\rangle|0\rangle$ to node $v$, and has node $v$ evaluate a quantum implementation of $\mathcal{I}_v$, which we will consider as a call to an oracle mapping $|x\rangle|0\rangle$ to $|x\rangle|\mathcal{I}_v(x)\rangle$ for all $x\in V$.
Node $v$ sends back the resulting qubit register, and node $u$ has evaluated $g_{uv}(|\psi\rangle)$ in 2 rounds.
Now, since $u$ can evaluate $g_{uv}$ in superposition, node $u$ may proceed using standard amplitude amplification, using 2 rounds of communication for each evaluation of $g_{uv}$, so that $u$ can find an element $w\in \mathcal{N}_G(u)$ satisfying $g_{uv}(w)=1$ with high probability in $\tilde{\mathcal{O}}(r\sqrt{|\mathcal{N}_G(u)|})$ rounds if one exists. We note that in this example, $v$ cannot execute this procedure by itself since it does not know $\mathcal{N}_G(u)$ (and sending this information to $v$ would take $|\mathcal{N}_G(u)|$ rounds), though it is able to evaluate $\mathcal{I}_v$ in superposition for any $w\in\mathcal{N}_G(u)$. 
For any classical procedure $\mathcal{C}$ evaluating a different function from this specific $g$ (that can be implemented efficiently classically and, therefore, translated to an efficient quantum implementation), the same idea results in the square-root advantage to find a desired element such that $g$ evaluates to $1$.

\subsection{Notation and Problem Definitions}\label{sec:notation and definitions}
For an integer-weighted graph $G = (V, E, W)$,
we will denote $n := |V|, m := |E|,$ and $W_e$ the weight of an edge $e \in E$ throughout the paper.
Let $\delta(v) \subset V$ be the set of edges incident on node $v$, and $\mathcal{N}_{G}(u) := \{v: uv \in E\}$ the neighborhood of $u\in G$.
Denote by $d_G(u,v)$ the shortest-path distance in $G$ from $u$ to $v$.
For a graph $G = (V,E,W)$ two sets of nodes $U$ and $U'$,
let $\mathcal{P}_G(U, U'):= \{uv\in E: u\in U, w\in U'\}$ be the set of edges connecting $U$ to $U'$.
Let $\mathcal{P}(U) := \mathcal{P}(U, U)$ as shorthand. All logarithms will be taken with respect to base 2, unless otherwise stated.

\begin{definition}[Steiner Tree Problem]
Given a weighted, undirected graph $G = (V, E, W)$, and a set of nodes $\mathcal{Z}\subset V$, referred to as \textit{Steiner Terminals}, output the minimum weight tree in $G$ that contains $\mathcal{Z}$. 
\end{definition}
\begin{definition}[Approximate Steiner Tree]
For a Steiner Tree Problem with terminals $\mathcal{Z}$ and solution $\mathcal{S}_{OPT}$ with edge set $E_{\mathcal{S}_{OPT}}$, a tree $T$ in $G$ containing $\mathcal{Z}$ with edge set $E_T$ such that 
$$
\sum_{uv\in E_T}W_{uv} \leq r\cdot \sum_{uv\in E_{\mathcal{S}_{OPT}}}W_{uv}
$$
is called an approximate Steiner Tree with approximation factor $r$. 
\end{definition}

\begin{definition}[Directed Minimum Spanning Tree Problem (DMST)]\label{def:DMST}Given a directed, weighted graph $G = (V, E, W)$ and a root node $r \in V$, output the minimum weight directed tree $T^*$ in $G$ such that there exists a directed path in $T^*$ from $r$ to any other node of $G$. This is also known as the {\em minimum weight arborescence} problem.  
\end{definition}

\section{Contributions}
We provide an algorithm for the qCCM that produces an approximate Steiner Tree with high probability (w.h.p.) in $\tilde{\mathcal{O}}(n^{1/4})$ rounds and an algorithm that produces an exact Directed Minimum Spanning Tree w.h.p. in $\tilde{\mathcal{O}}(n^{1/4})$ rounds. To do this, we enhance the quantum APSP algorithm of \cite{IzumiLeGallMag2019_APSP_QuantumDist} in an efficient way to compute not only APSP distances but also the corresponding routing tables (described in \S\ref{sec: APSP and routing}) that our algorithms rely on.
Further, in addition to these $\tilde{\mathcal{O}}$ results, in sections \ref{sec:complexity}, \ref{sec: Steiner alg complexity and correctness}, and \ref{sec:dmst complexity}, we characterize the constants and logarithmic factors involved in our algorithms as well as related classical algorithms to contribute to the community's understanding of their implementability. This reveals that the factors commonly obscured by $\tilde{\mathcal{O}}$ notation in related literature, especially the logarithms, have a severe impact on practicality. 
%Our results are mostly inspired by synthesizing the classical algorithms of \cite{Saikia_Karmakar2019_SteinerTree_CONGESTCLIQUE} and \cite{fischer2021_DMST} with the quantum techniques used in \cite{IzumiLeGallMag2019_APSP_QuantumDist}. 

We summarize the algorithmic results in the following two theorems: 
\begin{theorem} \label{thm: Steineralg}
There exists an algorithm in the Quantum CONGEST-CLIQUE model that, given an integer-weighted input graph $G = (V, E, W)$, outputs a $2(1-1/l)$ approximate Steiner Tree with probability of at least $1-\frac{1}{poly(n)}$, and uses $\tilde{\mathcal{O}}(n^{1/4})$ rounds of computation, where $l$ denotes the number of terminal leaf nodes in the optimal Steiner Tree.
\end{theorem}

\begin{theorem} \label{thm: DMSTalg}
There exists an algorithm in the Quantum CONGEST-CLIQUE model that, given a directed and integer-weighted input graph $G = (V, E, W)$, produces an exact Directed Minimum Spanning Tree with high probability, of at least $1-\frac{1}{poly(n)}$, and uses $\tilde{\mathcal{O}}(n^{1/4})$ rounds of computation. 
\end{theorem}

\section{APSP and Routing Tables} \label{sec: APSP and routing}
We first describe an algorithm for the APSP problem with routing tables in the qCCM, for which we combine an algorithm of \cite{IzumiLeGallMag2019_APSP_QuantumDist} with a routing table computation from \cite{witness_ref_Zwick2000}. For this, we reduce APSP with routing tables to triangle finding via \textit{distance products} as in \cite{CensorHillel2016_AlgebraicMethodsCCM_fast_APSP}.

\subsection{Distance Products and Routing Tables}\label{sec:distance products and triangle finding}
\begin{definition}\label{def: routing table}
A \textit{routing table} for a node $v$ is a function $R_v: V \to V$ mapping a vertex $u$ to the first node visited in the shortest path going from $v$ to $u$ other than $v$ itself.
\end{definition}

\begin{definition}\label{def: distance products}
The \textit{distance product} between two $n\cross n$ matrices $A$ and $B$ % also referred to as the min-plus or tropical product, 
is defined as the $n\cross n$ matrix $A \star B$ with entries:  
\begin{align}
    (A \star B)_{ij} = \min_k\{A_{ik} + B_{kj}\}.
\end{align}
\end{definition}
The distance product is also sometimes called the min-plus or tropical product.
For shortest paths, we will repeatedly square the graph adjacency matrix with respect to the distance product.
For an $n\cross n$ matrix $W$ and an integer $k$, let us denote
%\begin{align}
    $W^{k, \star} := W\star(W \star ( \dots(W\star W)) \dots)$
%\end{align}
as the $k^{th}$ power of the distance product.
For a graph $G = (V, E, W)$ with weighted adjacency matrix $W$ (assigning $W_{uv} = \infty$ if $uv \notin E$), $W^{k, \star}_{uv}$ is the length of the shortest path from $v$ to $u$ in $G$ using at most $k$ hops.
Hence, for any $N \geq n$, $W^{N, \star}$ contains all the shortest path distances between nodes in $G$.
As these distance products obey standard exponent rules, we may take $N = 2^{\ceil{\log n}}$ to recursively compute the APSP distances via taking $\ceil{\log n}$ distance product squares:
\begin{align}
    W^{2, \star} = W\star W, \phantom{x}
    W^{4,\star} = \left(W^{2, \star}\right)^{2,\star},
    \dots, \phantom{x}
    W^{2^{\ceil{\log n}}, \star} = \left(W^{2^{\ceil{\log n}-1}, \star} \right)^{2, \star}. \label{eq: DistProdSquare}
\end{align}
This procedure reduces computing APSP distances to computing $\ceil{\log n}$ distance products.
In the context of the CONGEST-CLIQUE model, each node needs to learn the row of $W^n$ that represents it. As we also require nodes to learn their routing tables, we provide a scheme in \S \ref{sec: routing via witness} that is well-suited for our setting to extend \cite{IzumiLeGallMag2019_APSP_QuantumDist} to also compute routing tables. 

\subsection{Distance Products via Triangle Finding}\label{sec: distance prod via triangles}

Having established reductions to distance products, we turn to their efficient computation.
The main idea is that we can reduce distance products to a binary search in which each step in the search finds negative triangles.
This procedure corresponds to \cite[Proposition~2]{IzumiLeGall2020_TriangleFinding}, which we describe here, restricting to finding the distance product square needed for Eq.~\eqref{eq: DistProdSquare}.
%for simplicity. 

A negative triangle in a weighted graph is a set of edges $\Delta^- = (uv, vw, wu) \subset E^3$ such that $\sum_{e\in \Delta^-} W_e < 0$. Let us denote the set of all negative triangles in a graph $G$ as $\Delta^-_G$.
Specifically, we will be interested in each node $v$ being able to output edges $vu \in \delta(v)$ such that $vu$ is involved in at least one negative triangle in $G$.
Let us call this problem \ref{prob:findedges}, and define it formally as: 

\begin{mdframed}\namedlabel{prob:findedges}{
\texttt{FindEdges}}
\begin{itemize}
    \item[] Input: An integer-weighted (directed or undirected) graph $G = (V, E, W)$ distributed among the nodes, with each node $v$ knowing $\mathcal{N}_G(v)$, as well as the weights $W_{vu}$ for each $u \in \mathcal{N}_G(v)$.
    \item[] Output: For each node $v$, its output is all the edges $vu \in E$ that are involved in at least one negative triangle in $G$.
\end{itemize}
\end{mdframed}

\begin{prop}\label{prop: FindEdges to DistProds}
If \ref{prob:findedges} on an $n$-node integer-weighted graph $G = (V, E, W)$ can be solved in $T(n)$ rounds, then the distance product $A \star B$ of two $n\cross n$ matrices $A$ and $B$ with entries in $[M]$ can be computed in $T(3n)\cdot \ceil{\log_2(2M)}$ rounds.
\end{prop}

\begin{proof}
 %Given an input graph $G = (v, E, W)$ we wish for each node $v$ to know $min_u \{W_{vu} + W_{uz}\}$ for each $z \in V$ for the distance product square computation. 
 Let $A$ and $B$ be arbitrary $n \cross n$ integer-valued matrices, and $D$ be an $n \cross n$ matrix initialized to $\0$. Let each $u\in V$ simulate three copies of itself,$u_1, u_2, u_3$, writing $V_1, V_2, V_3$ as the sets of copies of nodes in $V$. Consider the graph $G' = (V_1 \cup V_2 \cup V_3, E', W')$, by letting $u_iv_j \in E'$ for $u_i \in V_i, v_j \in V_j, i\neq j$, taking $W_{u_1v_2}' = A_{uv}$ for $u_1 \in V_1, v_{2} \in V_{2}$, 
  $W_{u_2v_3}' = B_{uv}$ for $u_2 \in V_2, v_{3} \in V_{3}$,
  and $W_{u_3v_1}' = D_{uv}$ for $u_3 \in V_3, v_1 \in V_{1}$.
%   A negative triangle in $G'$ corresponds to three edges such that 
% \begin{align}
% W'_{v_1u_2} + W'_{u
% _2z_3} < -D_{z_3v_1}, \label{eq: inequality for triangle search}
% \end{align}
% and hence, 
An edge $zv$ is part of a negative triangle in $G'$ exactly whenever
$$
\min_{u \in V} \{A_{vu} + B_{uz}\} < -D_{zv}.
$$
Assuming we can compute \ref{prob:findedges} for a $k$-node graph in $T(n)$ rounds, 
with a non-positive matrix $D = \mathbf{0}$ initialized we can apply simultaneous binary searches on $D_{zv}$, with values between $\{-2M, 0\}$, updating it for each node $v$ after each run of \ref{prob:findedges} to find 
$\min_{u \in V} \{A_{vu} + B_{uz}\}$ for every other node $z$ in\\ $T(3n)\cdot \ceil{\log(\max_{v, z \in V}\{\min_{u \in V}\{ A_{vu}+B_{uz}\}\})}$ rounds, since $G'$ is a tripartite graph with $3n$ nodes.
%We may add an additional round to check whether entries are $\infty$ by first setting $D$ to the most negative possible values before starting the binary search, potentially saving time on the binary search for those entries. 
\end{proof}

\begin{remark}
    This procedure can be realized in a single $n$-node distributed graph by letting each node represent the three copies of itself since $G'$ is tripartite. The $T(3n)$ stems from each processor node possibly needing to send one message for each node it is simulating in each round of \ref{prob:findedges}. If bandwidth per message is large enough (3 times the bandwidth needed for solving \ref{prob:findedges} in $T(n)$ rounds), then this can be done in $T(n)$ rounds. 
\end{remark}

So for this binary search, each node $v$ initializes and locally stores $D_{vz} = 0$ for each other $z \in V$, after which we solve \ref{prob:findedges} on $G'$.
The node then updates each $D_{vz}$ according to whether or not the edge copies of $vz$ were part of a negative triangle in $G'$, after which \ref{prob:findedges} is computed with the updated values for $D$.
This is repeated until all the $\min_{u \in V} \{A_{vu} + B_{uz}\}$ have been determined. \\

\subsection{Routing Tables via Efficient Computation of Witness Matrices}\label{sec: routing via witness}

For the routing table entries, we also need each node $v$ to know the intermediate node $u$ that is being used to attain $\min_{u \in V} \{W_{vu} + W_{uz}\}$. 
\begin{definition}
        For a distance product $A\star B$ of two $n \times n$ matrices $A, B$, a {\em witness matrix} $C$ is an $n\times n$ matrix such that 
        \begin{align*}
            C_{ij} \in argmin_{k\in [n]} \{A_{ik}+B_{kj}\}
        \end{align*}
    \end{definition}
Put simply, a witness matrix contains the intermediate entries used to attain the values in the resulting distance product.
We present here a simple way of computing witness matrices along with the distance product by modifying the matrix entries appropriately, first considered by \cite{witness_ref_Zwick2000}.
The approach is well-suited for our algorithm, as we only incur $\mathcal{O}(\log n)$ additional calls to \ref{prob:findedges} for a distance product computation with a witness matrix.

For an $n \cross n$ integer matrix $W$, obtain matrices $W'$ and $W''$ by taking $W'_{ij} = n W_{ij}+ j -1$ and $W''_{ji} = nW_{ji}$. Set $K = W'\star W''$.

\begin{claim} \label{claim: witness from dist}
With $W, W', W'',$ and $K$ as defined immediately above, 
\begin{itemize}
        \item[(i)] $\floor*{ {\dfrac{K}{n}} } = W^{2, \star}$
    \item[(ii)] $(K \mod n)+1$ is a witness matrix for $W^{2, \star}.$
\end{itemize}
\end{claim}
The claim follows from routine calculations of the quantities involved and can be found in the Appendix, \S \ref{app:witnessfromdist}.

Hence, we can obtain witness matrices by simply changing the entries of our matrices by no more than a multiplicative factor of $n$ and an addition of $n$. 
Since the complexity of our method depends on the magnitude of the entries of $W$ logarithmically, we only need logarithmically many more calls to \ref{prob:findedges} to obtain witness matrices along with the distance products, making this simple method well-suited for our approach.
%This provides an improvement in our witness matrix computation protocol compared to \cite{CensorHillel2016_AlgebraicMethodsCCM_fast_APSP}, applied in both \cite{Saikia_Karmakar2019_SteinerTree_CONGESTCLIQUE} and \cite{IzumiLeGallMag2019_APSP_QuantumDist}, which uses up to $O(\log^3 n)$ calls to the distance product to obtain a witness matrix for a single distance product.
More precisely, we can compute $W^{2, \star}$ with a witness matrix using $\ceil*{\log(2n\cdot \max_{i,j} \{W^{2,\star}_{ij} < \infty\})}.$ calls to \ref{prob:findedges}.
%, where the $+1$ comes from first checking whether entries are $\infty$ before starting the binary search.
%in this full version should probably talk in detail talk about how we can deal with the \infty edges
We obtain the following corollary to proposition \ref{prop: FindEdges to DistProds} to characterize the exact number of rounds needed: 
\begin{corollary} \label{cor: distance product from findedges}
If \ref{prob:findedges} on an $n$-node integer-weighted graph $G = (V, E, W)$ can be solved in $T(n)$ rounds, then the distance product square $W^{2, \star}$, along with a witness matrix $H$, can be computed in $T(3n)\cdot \ceil{\log_2(n\cdot \max_{v, z \in V}\{\min_{u \in V}\{ W_{vu}+W_{uz}\}\}+n)}$ rounds.
\end{corollary}
\begin{proof}
This follows from claim \ref{claim: witness from dist} and proposition \ref{prop: FindEdges to DistProds} upon observing that\\ $\max_{v, z \in V}\{\min_{u \in V}\{ W'_{vu}+W''_{uz})\}\} \leq n\cdot \max_{v, z \in V}\{\min_{u \in V}\{ W_{vu}+W_{uz}\}\}+n$.
\end{proof}

Once we obtain witness matrices along with the distance product computations, constructing the routing tables for each node along the way of computing APSP is straightforward.
In each squaring of $W$ in Eq.~\eqref{eq: DistProdSquare}, each node updates its routing table entries according to the corresponding witness matrix entry observed.
It is worth noting that these routing table entries need only be stored and accessed classically so that we avoid using unnecessary quantum data storage.

\subsection{Triangle Finding} \label{sec:trangle finding}

Given the results from sections \ref{sec: routing via witness} and \ref{sec: distance prod via triangles}, we have reduced finding both the routing tables and distance product to having each edge learn the edges involved in a negative triangle in the graph. 
This section will thus describe the procedure to solve the \ref{prob:findedges} subroutine.
We state here a central result from \cite{IzumiLeGallMag2019_APSP_QuantumDist}:

\begin{prop} \label{thm: triangle edge finding}
There exists an algorithm in the quantum CONGEST-CLIQUE model that solves the \ref{prob:findedges} subroutine in $\tilde{\mathcal{O}}(n^{1/4})$ rounds. 
\end{prop}
We will proceed to describe each step of the algorithm to describe the precise round complexity beyond the $\tilde{O}(n^{1/4})$ to characterize the constants involved in the interest of assessing the future implementability of our algorithms. \\

As a preliminary, we give a message routing lemma of \cite{Dolev_lenzen_2012TriTA} for the congested clique, which will be used repeatedly:
\begin{lemma}\label{lemma: dolev routing lemma}
Suppose each node in $G$ is the source and destination for at most $n$ messages of size $\mathcal{O}(\log n)$ and that the sources and destinations of each message are known in advance to all nodes.
Then all messages can be routed to their destinations in $2$ rounds. 
\end{lemma}

We introduce the subproblem \texttt{FindEdgesWithPromise} (\ref{prob:FEWP} henceforth).
Let $\Gamma(u,v)$ denote the number of nodes $w\in V$ such that $(u,v,w)$ forms a negative triangle in $G$.
\begin{mdframed}\namedlabel{prob:FEWP}{\texttt{FEWP}}
% \begin{itemize}
    \item[] Input: An integer-weighted graph $G = (V, E, W)$  distributed among the nodes and a set $S\subset\mathcal{P}(V)$, with each node $v$ knowing $\mathcal{N}_G(v)$ and $S$.
    \item[] Promise: For each $uv \in S, \Gamma(u,v)\leq 90\log n.$ 
    \item[] Output: For each node $v$, its output is the edges $vu \in S$ that satisfy $\Gamma(u,v) > 0$.
% \end{itemize}
\end{mdframed}

%\begin{critic}
%$\bullet$ write the pseudocode of the algo for \ref{prob:findedges} using FEWP here, probably leave the proof of correctness as a reference to \cite{IzumiLeGallMag2019_APSP_QuantumDist}?
%\end{critic}

We give here a description of the procedure of \cite{IzumiLeGallMag2019_APSP_QuantumDist} to solve \ref{prob:findedges} given an algorithm $\mathcal{A}$ to solve \ref{prob:FEWP}. Let $\eps_{\mathcal{A}}$ be the failure probability of the algorithm $\mathcal{A}$ for an instance of \ref{prob:FEWP}. 
\begin{mdframed}
\namedlabel{alg:FindEdgesViaFEWP}{\texttt{FindEdgesViaFEWP}}
\begin{itemize}
    \item[\namedlabel{itm:FindTris1}{1}:] $S := \mathcal{P}; M := \emptyset; i := 0$. 
    
    \item[\namedlabel{itm:FindTris2}{2}:] WHILE $60\cdot 2^i \log n \leq n$: 
    \begin{itemize}
        \item[a):] Each node samples each of its edges with probability $\sqrt{\frac{60\cdot 2^i\log n}{n}}$, so that we obtain a distributed subgraph $G'$ of $G$ consisting of the sampled edges
        \item[b):] Run $\mathcal{A}$
        on $(G', S)$. Denote the output by $S'$. 
        \item[c):] $S \leftarrow S\setminus S'; M \leftarrow M\cup S; i \leftarrow i+1.$
    \end{itemize}
    
    \item [\namedlabel{itm:FindTris3}{3}:]Run $\mathcal{A}$ on $(G,S)$, and call $S''$ the output. 
    \item [\namedlabel{itm:FindTris4}{4}:]Output $M\cup S$. 
\end{itemize}
\end{mdframed}
 From step \ref{itm:FindTris2} of this above algorithm, it is straightforward to check that this requires a maximum of \\$c_n := \ceil{\log\text{(}\frac{n}{60\log n}\text{)}}+1$ calls to the $\mathcal{A}$ subroutine to solve \ref{prob:FEWP}. 
Further, it succeeds with probability at least $1-c_n/n^3 -c_n/n^28-(c_n+1)\varepsilon_{\mathcal{A}}$. We refer the reader to \cite[\S 3]{IzumiLeGallMag2019_APSP_QuantumDist} for the proof of correctness.
We now turn toward constructing an efficient algorithm for FEWP. 

To solve this subroutine, we must first introduce an additional labeling scheme over the nodes that will determine how the search for negative triangles will be split up to avoid communication congestion in the network.
Assume for simplicity that $n^{1/4}, \sqrt{n}, n^{3/4}$ are integers. 
Let $\mathcal{M} = [n^{1/4}] \times [n^{1/4}] \times [\sqrt{n}]$. 
Clearly, $|\mathcal{M}| = n$, and $\mathcal{M}$ admits a total ordering lexicographically. 
Since we assume each node $v_i \in V$ is labeled with unique integer ID $i \in [n]$, $v_i$ can select the element in $\mathcal{M}$ that has place $i$ in the lexicographic ordering of $\mathcal{M}$ without communication occurring. 
Hence, each node $v\in V$ is associated with a unique triple $(i,j,k) \in \mathcal{M}$.
We will refer to the unique node associated with $(i,j,k) \in \mathcal{M}$ as node $v_{(i,j,k)}$.\\
The next ingredient is a partitioning scheme of the space of possible triangles. 
Let $\mathcal{U}$ be a partition of $V$ into $n^{1/4}$ subsets containing $n^{3/4}$ nodes each,
by taking $U_i := \{v_j: j\in \{(i-1)\cdot n^{3/4}, \dots, i\cdot n^{3/4}\}\}$ for $i=1, \dots, n^{1/4}$, and $\mathcal{U}:= \{U_1, \dots, U_{n^{1_4}}\}$.
Apply the same idea to create a partition $\mathcal{U}'$ of $\sqrt{n}$ sets of size $\sqrt{n}$, by taking $U_i' := \{v_j: j\in \{(i-1)\cdot \sqrt{n}, \dots, i\cdot \sqrt{n}\}\}$ for $i=1, \dots, \sqrt{n}$, and $\mathcal{U}:= \{U_1, \dots, U_{\sqrt{n}}\}$.
Let $\mathbb{V} = \mathcal{U}\times \mathcal{U} \times \mathcal{U}'$. Each node $v_{(i,j,k)}$ can then locally determine its association with the element $(U_i, U_j, U_k') \in \mathbb{V}$ since $|\mathbb{V}|= n$. 
Further, if we use one round to have all nodes broadcast their IDs to all other nodes, each node $v_{(i,j,k)}$ can locally compute the $(U_i, U_j, U_k')$ it is assigned to, so this assignment can be done in one round. 

We present here the algorithm \ref{alg:computepairs} used to solve the FEWP subroutine. 
\begin{mdframed}\namedlabel{alg:computepairs}{\texttt{ComputePairs}}
% \begin{itemize}
    \item[] Input: An integer-weighted graph $G = (V, E, W)$  distributed among the nodes, a partition of $V\times V\times V$ of $(U_i, U_j, U'_k)$ associated with each node as above, and a set $S\subset\mathcal{P}(V)$ such that for $uv \in S, \Gamma(u,v)\leq 90\log n.$ 
    \item[] Output: For each node $v$, its output is the edges $vu \in S$ that satisfy $\Gamma(u,v) > 0$.
% \end{itemize}

\begin{itemize}
    \item [\namedlabel{itm:CP1}{1}:]Every node $v_{(i, j, k)}$ receives the weights $W_{uv}$, $W_{vw}$ for all $uv \in \mathcal{P}(U_i, U_j)$ and $vw \in \mathcal{P}(U_j, U_k')$. 
    \item [\namedlabel{itm:CP2}{2}:]Every node $v_{(i,j,k)}$ constructs the set $\Lambda_k(U_i, U_j) \subset \mathcal{P}(U_i, U_j)$ by selecting every $uv \in \mathcal{P}(U_i, U_j)$ with probability $10\cdot \frac{\log n}{\sqrt{n}}$. 
    If $|\{v\in U_1: uv \in \Lambda_k(U_i, U_j)\}| > 100n^{1/4} \log n$ for some $u\in U_j$, abort the algorithm and report failure. 
    Otherwise, $v_{(i, j, k)}$ keeps all pairs $uv \in \Lambda_k(U_i, U_j) \cap S$ and receives the weights $W{uv}$ for all of those pairs. 
    Denote those elements of $\Lambda_k(U_i, U_j) \cap S$ as $u_1^k v_1^k, \dots, u_m^k v_m^k$. 
    \item [\namedlabel{itm:CP3}{3}:]Every node $v_{(i,j,k)}$ checks for each $l\in [m]$ whether there is some $U\in \mathcal{U}'$ that contains a node $w$ such that $(u^k_l, v^k_l, w)$ forms a negative triangle, 
    and outputs all pairs $u^k_l v^k_l$ for which a negative triangle was found. 
\end{itemize}
\end{mdframed}

With probability at least $1-2/n$, the algorithm \ref{alg:computepairs} does not terminate at step \ref{itm:CP2} and every pair $(u,v) \in S$ appears in at least one $\Lambda_k(U_i, U_j)$. 
The details for this result can be found in \cite[Lemma 2]{IzumiLeGallMag2019_APSP_QuantumDist}. 

Step \ref{itm:CP1} requires $2n^{1/4}\ceil{\frac{\log W}{\log n}}$ rounds and can be implemented fully classically without any qubit communication. 
Step \ref{itm:CP2} requires at most $200 \log n \ceil{\frac{\log W}{\log n}}$ rounds and can also be implemented classically. 
Step \ref{itm:CP3} can be implemented in $\tilde{\mathcal{O}}(n^{1/4})$ rounds quantumly taking advantage of distributed Grover search but would take $\mathcal{O}(\sqrt{n})$ steps to implement classically. The remainder of this section is devoted to illustrating how this step can be done in $\tilde{\mathcal{O}}(n^{1/4})$ rounds.

Define the following quantity:
\begin{definition}\label{def: delta}
For node $v_{(i, j, k)}$, let 
$$
\Delta(i, j, k):= \{(u, v) \in \mathcal{P}(U_i, U_j)\cap S: \exists w \in U_k' \textit{ with } (u, v, w) \textit{ forming a negative triangle in } G\} 
$$
\end{definition}

For simultaneous quantum searches, we divide the nodes into different classes based on the number of negative triangles they are a part of with the following routine:

\begin{mdframed}\namedlabel{alg:IdentifyClass}{\texttt{IdentifyClass}}
% \begin{itemize}
    \item[] Input: An integer-weighted graph $G = (V, E, W)$  distributed among the nodes, and a set $S\subset E$ as in \ref{prob:FEWP}. 
    \item[] Output: For each node $v$, a class $\alpha$ the node belongs to.
% \end{itemize}
\begin{itemize}
    \item [\namedlabel{itm:IC1}{1}:] Every node $u_{(i,j,k)}\in V$ samples each node in $\{v\in V: (u_{(i,j,k)},v)\in S\}$ with probability $\frac{10\log n}{n}$, creating a set $\Lambda(u)$ of sampled vertices. If $\max_u |\Lambda(u)| > 20 \log n$, abort the algorithm and report a failure. 
    Otherwise, have each node broadcast $\Lambda(u)$ to all other nodes, and take $R := \cup_{u\in V}\{uv|v\in \Lambda(u)\}$. 
    \item [\namedlabel{itm:IC2}{2}:]
    Each $v_{(i,j,k)} \in V$ computes $d_{i,j,k} := |\{  uv\in \mathcal{P}(i,j)\cap R:\phantom{x} \exists w \in U'_k$ such that $\{u,v,w\}$  forms a negative triangle in $G \}|$, then determines its class $\alpha$ to be $min\{c\in \N:d_{i,j,k} < 10\cdot 2^c\log n\}$.
    
\end{itemize}
\end{mdframed}
This uses at most $20\log n$ rounds (each node sends at most that many IDs to every other node) and can be implemented by having all exchanged messages consist only of classical bits. 
%Further, for this procedure, allowing constantly more communication bandwidth helps complexity directly, e.g., if we allow $C\cdot \log n$ bits to be sent each round, we can implement \ref{alg:sampleneighbors} in $\ceil{\frac{20\log n}{C}}$ rounds. 
Using Chernoff's bound, one can show that the procedure succeeds with probability of at least $1-1/n$ as seen in \cite[Proposition 5]{IzumiLeGall2020_TriangleFinding}.

Let us make the convenient assumption that $\alpha = 0$ for all $v_{i,j,k}$, which avoids some technicalities around congestion in the forthcoming triangle search. Note that $\alpha \leq \frac{1}{2}\log n$, so we can run successive searches for each $\alpha$ for nodes in with class $\alpha$ in the general case. The general case is discussed in \S \ref{app:alpha nonzero case} of the appendix and can also be found in \cite{IzumiLeGallMag2019_APSP_QuantumDist}, but this case is sufficient to convey the central ideas.

%I'm deciding to skip some things here, such as proposition 5, and lemmas 3 and 4. 
%We can discuss whether we want those details or leave them to the reference and just mention the parts of them that we need to describe the algorithm. 
We have all the necessary ingredients to describe the implementation of step \ref{itm:CP3} of the \ref{alg:computepairs} procedure. 
\begin{mdframed}
\begin{itemize}
    \item [\namedlabel{itm:CP3.1}{3.1}:] Each node executes the \ref{alg:IdentifyClass} procedure. 
    \item [\namedlabel{itm:CP3.2}{3.2}:]
    For each $\alpha$, for every $l\in [m]$, every node $v_{(i,j,k)}$ in class $\alpha$ executes a quantum search to find whether there is a $U_k'\in \mathcal{U}'$ with some $w\in U_k'$ forming a negative triangle $(u^k_l, v^k_l, w)$ in $G$, and then reports all the pairs $u^k_l v^k_l$ for which such a $U_k'$ was found.
\end{itemize}
\end{mdframed}
%describe the quantum Grover search and the alphas, then this part is done. 

%feeling a bit unsure how much of these I should prove. The proofs are in Izumi Le Gall and are technical, and I can put them in here but not sure how much value they bring. 
This provides the basis of the triangle-searching strategy. To summarize the intuition of the asymptotic speedup in this paper: Since the $U'_k$ have size $\sqrt{n}$ (recall that $|\mathcal{U}'| = \sqrt{n}$), if each node using a quantum search can search through its assigned $U'_k$ in $\tilde{\mathcal{O}}(n^{1/4})$ rounds, simultaneously, we will obtain our desired complexity. We will complete this argument in \S \ref{sec:finishing triangles} and first describe the quantum searches used therein in the following subsection. 

\subsection{Distributed Quantum Searches}\label{sec: distributed quantum search}

With this intuition in mind, we now state two useful theorems of \cite{IzumiLeGallMag2019_APSP_QuantumDist} for the distributed quantum searches. Let $X$ denote a finite set throughout this subsection.
\begin{theorem}
Let $g:X\rightarrow \{0,1\}$, if a node $u$ can compute $g(x)$ in $r$ rounds in the CONGEST-CLIQUE model for any $x\in X$, then there exists an algorithm in the Quantum CONGEST-CLIQUE that has $u$ output some $x\in X$ with $g(x) = 1$ with high probability using $\tilde{\mathcal{O}}(r\sqrt{|X|})$ rounds. 
\end{theorem}

This basic theorem concerns only single searches, but we need a framework that can perform multiple simultaneous searches. Let $g_1, \dots,g_m: X \rightarrow \{0, 1\}$ and 
$$A_i^0 := \{x\in X: g_i(x) = 0\}, A_i^1 := \{x\in X: g_i(x) = 1\}, \forall i \in [m].$$
Assume there exists an $r$-round classical distributed algorithm $C_m$ that allows a node $u$ upon an input $\chi = (x_1, \dots, x_m) \in X^m$ to determine and output $(g_1(x_1), \dots, g_m(x_m))$.
In our use of distributed searches, $X$ will consist of nodes in the network, and searches will need to communicate with those nodes for which the functions $g_i$ are evaluated. 
To avoid congestion, we will have to consider those $\chi \in X^m$ that have many repeated entries carefully. We introduce some notation for this first. Define the quantity 
$$\alpha(\chi) := \max_{I\subset[m]}|\{ \chi_i = \chi_j \quad \forall i, j \in I\}|,$$ the maximum number of entries in $\chi$ that are all identical. 

Next, given some $\beta \in \N$, assume that in place of $C_m$ we now have a classical algorithm $\tilde{C}_{m, \beta}$ such that upon input $\chi = (x_1, \dots, x_m) \in X^m$, a node $u$ outputs $g_1(x_1), \dots, g_m(x_m)$ if $\alpha(\chi) \leq \beta$ and an arbitrary output otherwise. 
The following theorem summarizes that such a $\tilde{C}_{m, \beta}$ with sufficiently large $\beta$ is enough to maintain a quantum speedup as seen in the previous theorem: 

\begin{theorem}\label{thm: distquantsearches}
For  a set $X$ with $|X|<m/(36\log m)$, suppose there exists such an evaluation algorithm $C_{m, \beta}$ for some $\beta>8m/|X|$ and that $\alpha(\chi) \leq \beta$ for all $\chi \in A_1^1\times \cdot \cdot \cdot \times A_m^1$. Then there is a $\tilde{\mathcal{O}}(r\sqrt{|X|})$-round quantum algorithm that outputs an element of $A_1^1\times \cdot\cdot\cdot \times A_m^1$ with probability at least $1-2/m^2$. 
\end{theorem}
The proof can be found in \cite[Theorem 3]{IzumiLeGallMag2019_APSP_QuantumDist}.

\subsection{Final Steps of the Triangle Finding}\label{sec:finishing triangles}
We continue here to complete the step \ref{itm:CP3.2} of the \ref{alg:computepairs} procedure, armed with Theorem \ref{thm: distquantsearches}. We need simultaneous searches to be executed by each node $v_{(i,j,k)}$ to determine the triangles in $U_i\times U_j \times U'_k$. We provide a short lemma first that ensures the conditions for the quantum searches:
\begin{lemma}\label{lemma:IClemma}
The following statements hold with probability at least $1-2/n^2$: 
\begin{itemize}
    \item [\namedlabel{itm:IClemmaii}{(i)}:] $|\Delta(i,j,k)| \leq 2n$
    \item [\namedlabel{itm:IClemmaiv}{(ii)}:] $|\Lambda_k(U_i,U_j)\cap \Delta(i,j,k)|\leq 100\cdot \sqrt{n}\log n$ for $i,j\in [n^{1/4}]$. 
\end{itemize}
\end{lemma}
The proofs of these statements are technical but straightforward, making use of Chernoff's bound and union bounds; hence we skip them here. To invoke Theorem \ref{thm: distquantsearches}, we describe a classical procedure first, beginning with an evaluation step, \ref{alg:EvaluationA} implementable in $\tilde{\mathcal{O}}(1)$ rounds.

\begin{mdframed}\namedlabel{alg:EvaluationA}{\texttt{EvaluationA}} \\
Input: Every node $v_{(i,j,k)}$ receives $m$ elements $(u^{i,j,k}_1, \dots, u^{i,j,k}_m)$ of $\mathcal{U}'$\\
Promise: For every node $v_{i,j,k}$ and every $\w \in \mathcal{U}', |L_{\w}^{i,j,k}| \leq 800 \sqrt{n}\log n$. \\
Output: Each node outputs a list of exactly those $u^{i,j,k}_l$ such that there is a negative triangle in $U_i \times U_j \times u^{i,j,k}_l$.

\begin{enumerate}
    \item \label{itm:evalA1} Every node $v_{(i,j,k)}$, for each $r\in \sqrt{n}$ routes the list $L_{\w}^{i,j,t}$ to node $v_{(i,j,t)}$. 
    \item \label{itm:evalA2} Every node $v_{(i,j,k)}$, for each $vu$ it received in step \ref{itm:eval1}, sends the truth value of the inequality 
    \begin{align}
        \min_{w\in U'_k}\{W_{uw} + W_{wv}\} \leq W_{vu} \label{eq:feval}
    \end{align}
    to the node that sent $vu$. 
\end{enumerate}
\end{mdframed}

Each node is the source and destination of up to $800n\log n$ messages in step \ref{itm:evalA1}, meaning that this step can be implemented in $1600\log n$ rounds. The same goes for step \ref{itm:evalA2}, noting that the number of messages is the same, but they need only be single-bit messages (the truth values of the inequalities). Hence, the evaluations for Theorem \ref{thm: distquantsearches} can be implemented in $3200\log n$ rounds. Now, applying the theorem with $X = \mathcal{U}', \beta = 800\sqrt{n}\log n$, noting that then the assumptions of the theorem hold with probability at least $1-2/n^2$ due to Lemma \ref{lemma:IClemma}, implies that step \ref{itm:CP3.2} is implementable in $\tilde{\mathcal{O}}(n^{1/4})$ rounds, with a success probability of at least $1-2/m^2$. 

For the general case in which we do not assume $\alpha = 0$ for all $i,j,k$ in \ref{alg:IdentifyClass}, covered in the appendix, one needs to modify the \ref{alg:EvaluationA} procedure in order to implement load balancing and information duplication to avoid congestion in the simultaneous searches. These details can be found in the appendix, where a new labeling scheme and different evaluation procedure \ref{alg:EvaluationB}, are described for this, or in \cite{IzumiLeGallMag2019_APSP_QuantumDist}.

\subsection{Complexity}\label{sec:complexity}
As noted previously and in \cite{IzumiLeGallMag2019_APSP_QuantumDist}, this APSP scheme uses $\tilde{\mathcal{O}}(n^{1/4})$ rounds. Let us characterize the constants and logarithmic factors involved to assess this algorithm's practical utility. Suppose that in each round, $2\cdot \log n$ qubits can be sent in each message (so that we can send two IDs or one edge with each message), where $n$ is the number of nodes. For simplicity, let's assume that for the edge weights $W \ll n$ and drop $W$.
\begin{enumerate}
    \item APSP with routing tables needs $\log(n)$ distance products with witness matrices.
    \item Computing the $i^{th}$ distance product square for Eq.~\eqref{eq: DistProdSquare} with a witness matrix needs up to $\log(2^i)=i$ calls to \ref{prob:findedges}, since the entries of the matrix being squared may double each iteration. Then APSP and distance products together make $\sum_{i=1}^{\ceil{\log n}} i = \frac{\ceil{\log(n)}(\ceil{\log(n)}+1)}{2}$ calls to \ref{prob:findedges}.
    \item Solving \ref{prob:findedges} needs $\log (\frac{n}{60 \log n})$ calls to \ref{prob:FEWP}, using \ref{alg:FindEdgesViaFEWP}. 
    \item Step \ref{itm:CP1} of \ref{alg:computepairs} needs up to $2\cdot n^{1/4}$ rounds and step \ref{itm:CP2} takes up to $200 \log n$ rounds.
    \item Step \ref{itm:IC1} of \ref{alg:IdentifyClass} needs up to $20\log n$ rounds. 
    \item In step \ref{itm:IC2} of \ref{alg:IdentifyClass}, the $c_{uvw}$ are up to $\frac{1}{2}\log n$ large, and hence $\alpha$ may range up to $\frac{1}{2}\log n$.
    \item Step \ref{itm:eval0} of the \ref{alg:EvaluationB} procedure needs $n^{1/4}$ rounds. Steps \ref{itm:eval1} and \ref{itm:eval2} of the \ref{alg:EvaluationB} (or \ref{alg:EvaluationA}, in the $\alpha = 0$ case) procedure use a total of $3200\log n$ rounds.
    \item \ref{alg:EvaluationB} (or \ref{alg:EvaluationA}) procedure is called up to $\log(n) n^{1/4}$ times for each value of $\alpha$ in step \ref{itm:CP3.2} of \ref{alg:computepairs}. 
\end{enumerate}

Without any improvements, we get the following complexity, using $3n$ in place of $n$ for the terms of steps 3-8 due to corollary \ref{cor: distance product from findedges}: 
\begin{align} 
\nonumber
   \frac{\ceil{\log(n)}(\ceil{\log(n)}+1)}{2}   % calls to FindEdges
\log(\frac{3n}{60 \log 3n}) %FE to FEWP
\Big(2(3n)^{1/4} + %step 1 CP
    %  200\log n + %step 2 CP
    %  20\log n + %IC step 1
    220 \log 3n + %consolidating CP step 2 and IC step 1
     2(3n)^{1/4} + \\ %step 0 eval
     \frac{1}{2}\log 3n \cdot %each alpha
     \log 3n \cdot (3n)^{1/4}
     3200(\log 3n) %eval steps 1, 2
     \Big), \label{eq:quantum apsp complexity}
\end{align}

which we will call $f(n)$, so that $f(n) = \mathcal{O}(n^{1/4}\log^6(n))$, with the largest term being about $800 \log^6(n)n^{1/4}$,
and we have dropped $W$ to just consider the case $W \ll n$.
We can solve the problem trivially in the (quantum or classical) CONGEST-CLIQUE within $n\log(W)$ rounds by having each node broadcast its neighbors and the weight on the edge. Let us again drop $W$ for the case $W \ll n$ so that in order for the quantum algorithm to give a real speedup, we will need 
$$
f(n) < n,
$$
which requires $n>10^{18}$ (even with the simpler under-approximation $800 \log^6(n)n^{1/4}$ in place of $f$). Hence, even with some potential improvements, the algorithm is impractical for a large regime of values of $n$ even when compared to the trivial CONGEST-CLIQUE $n$-round strategy.\\
For the algorithm of \cite{IzumiLeGallMag2019_APSP_QuantumDist} computing only APSP \textit{distances}, the first term in \ref{eq:quantum apsp complexity} becomes simply $\ceil{\log n}$, so that when computing only APSP distances the advantage over the trivial strategy begins at roughly $n\approx 10^{16}$.
\begin{remark}
    In light of logarithmic factors commonly being obscured by $\tilde{\mathcal{O}}$ notation, we point out that even an improved algorithm needing only $\log^4(n) n^{1/4}$ would not be practical unless $n>10^7$, for the same reasons. Recall that $n$ is the number of {\em processors} in the distributed network -- tens of millions would be needed to make this algorithm worth implementing instead of the trivial strategy. Practitioners should mind the $\tilde{\mathcal{O}}$ if applications are of interest, since even relatively few logarithmic factors can severely limit practicality of algorithms, and researchers should be encouraged to fully write out the exact complexities of their algorithms for the same reason.  
\end{remark}

\subsubsection{Memory Requirements}
Although in definition \ref{def: QCCM} we make no assumption on the memory capacities of each node, the trivial $n$-round strategy uses at least $2\log(n)|E|^2\cdot\log(W)$ memory at the leader node that solves the problem. For the APSP problem in question, using the Floyd-Warshall algorithm results in memory requirements of $2n^2\log(n)\cdot\log(nW)$ at the leader node. Hence, we may ask whether the quantum APSP algorithm leads to lower memory requirements. The memory requirement is largely characterized by up to $720n^{7/4}\log(n)\log(nW)$ needed in step \ref{itm:eval0} of the \ref{alg:EvaluationB} procedure, which can be found in the appendix. This results in a memory advantage for quantum APSP over the trivial strategy beginning in the regime of $n>1.6\cdot 10^{10}$. 

\subsubsection{Complexity of the Classical Analogue}\label{sec: classical apsp complexity}
For completeness, we provide here a characterization of the complexity of a closely related classical algorithm for APSP with routing tables in the CONGEST-CLIQUE as proposed in \cite{CensorHillel2016_AlgebraicMethodsCCM_fast_APSP} that has complexity $\tilde{\mathcal{O}}(n^{1/3})$. In their framework, the approach to finding witness matrices requires $\mathcal{O}(\log(n)^3)$ calls to the distance product \cite[\S 3.4]{CensorHillel2016_AlgebraicMethodsCCM_fast_APSP}, and similarly to our approach $\log(n)$ distance products are required. Their classical algorithm computes distance products in $\mathcal{O}(n^{1/3})$ rounds, or under $2\log n$ message bandwidth in up to
    
\begin{align}\label{eq:classical apsp complexity}
  20n^{1/3}\log(n)^4 =: g(n)  
\end{align}

rounds, the details of which can be found in the appendix, \S  \ref{sec: appendix: classical apsp complexity details}. Then $g(n) > n$ up until about $n \approx 2.6 \cdot 10^{11}$. As with the quantum APSP, though this algorithm gives the best known asymptotic complexity of $\tilde{\mathcal{O}}(n^{1/3})$ in the classical CONGEST-CLIQUE, it also fails to give any real improvement over the trivial strategy across a very large regime of values of $n$. Consequently, algorithms making use of this APSP algorithm, such as \cite{Saikia_Karmakar2019_SteinerTree_CONGESTCLIQUE} or \cite{fischer2021_DMST}, suffer from the same problem of impracticality. However, the algorithm only requires within $4n^{4/3}\log(n)\log(nW) + n\log(n)\log(nW)$ memory per node, which is less than required for the trivial strategy even for $n\geq4$.

\section{Approximately Optimal Steiner Tree Algorithm}\label{main algo}

\subsection{Algorithm Overview}\label{sec: steiner algorithm overview}
We present a high-level overview of the proposed algorithm to produce approximately optimal Steiner Trees, divided into four steps. 

\begin{description}
    \item[Step \namedlabel{itm:apsp}{1} - APSP and Routing Tables:] Solve the APSP problem as in \cite{IzumiLeGallMag2019_APSP_QuantumDist} and add an efficient routing table scheme via triangle finding in $\tilde{\mathcal{O}}(n^{1/4})$ rounds, with success probability $(1-1/poly(n))$ (this step determines the algorithm's overall success probability). 
    \item[Step \namedlabel{itm:spf}{2} - Shortest-path Forest:] Construct a shortest-path forest (SPF), where each tree consists of exactly one source terminal and the shortest paths to the vertices whose closest terminal is that source terminal.
    This step can be completed in one round and $n$ messages, per \cite[\S 3.1]{Saikia_Karmakar2019_SteinerTree_CONGESTCLIQUE}.
    The messages can be in classical bits.
    \item[Step \namedlabel{itm:weights}{3} - Weight Modifications:] Modify the edge weights depending on whether they belong to a tree (set to 0), connect nodes in the same tree (set to $\infty$), or connect nodes from different trees (set to the shortest path distance between root terminals of the trees that use the edge). This uses one round and $n$ messages. 
    \item[Step \namedlabel{itm:mst}{4} - Minimum Spanning Tree:] Construct a minimum spanning tree (MST) on the modified graph in $\mathcal{O}(1)$ rounds as in \cite{MST_in_O1_CCM_Nowicki2019}, and prune leaves of the MST that do not connect terminal nodes since these are not needed for the Steiner Tree. 
\end{description}
The correctness of the algorithm follows from the correctness of each step together with the analysis of the classical results of \cite{Kou1981_fast_algo_for_ST_}, which uses the same algorithmic steps of constructing a shortest path forest and building it into an approximately optimal Steiner Tree.

\subsection{Shortest Path Forest}\label{sec: SPF}
After the APSP distances and routing tables have been found, we construct a \textit{Shortest Path Forest} (SPF) based on the terminals of the Steiner Tree. 
\begin{definition}{(Shortest Path Forest):}\label{def: SPF}
For a weighted, undirected graph $G = (V, E, W)$ together with a given set of terminal nodes $Z = \{z_1, \dots, z_k\}$, a subgraph $F = (V, E_F, W)$ of $G$ is called a \textit{shortest path forest} if it consists of $|Z|$ disjoint trees $T_z = (V_z, E_z, W)$ satisfying
\begin{itemize}
    \item[i)]$z_i \in T_{z_j}$ if and only if $i =j$, for $i, j \in [k]$.
    \item[ii)] For each $v\in Z_i, d_G(v, z_i) = \min_{z \in Z}d_G(v,z)$, and a shortest path connecting $v$ to $z_i$ in $G$ is contained in $T_{z_i}$
    \item[iii)] The $V_{z_i}$ form a partition of $V$, and $E_{z_1} \cup E_{z_2} \dots \cup E_{z_k}=E_F \subset E$
\end{itemize}
\end{definition}

In other words, an SPF is a forest obtained by gathering, for each node, a shortest path in $G$ connecting it to the closest Steiner terminal node.

%We now describe how to obtain such an SPF after having computed the APSP distances and routing tables. 
For a node $v$ in a tree, we will let $par(v)$ denote the parent node of $v$ in that tree, $s(v)$ the Steiner Terminal in the tree that $v$ will be in, and $ID(v)\in [n]$ the ID of node $v \in V$. Let $\mathcal{Q}(v):= \{z: d_G(v, z) = \min_{z\in Z} d_G(v,z)\}$ be the set of Steiner Terminals closest to node $v$. 
%Again, we may assume that the IDs of the nodes in the network are in $[n]$. 
We make use of the following procedure for the SPF: 
\begin{mdframed}\namedlabel{alg:spf}{\texttt{DistributedSPF}}
% \begin{itemize}
\item[] Input: For each node $v\in G$, APSP distances and the corresponding routing table $R_v$.
\item[] Output: An SPF distributed among the nodes. 
% \end{itemize}
\begin{itemize}
    \item [\namedlabel{itm:SPF1}{1}:] Each node $v$ sets $s(v) := \text{argmin}_{z\in \mathcal{Q}(v)} ID(z)$ using the APSP information.
    \item [\namedlabel{itm:SPF2}{2}:] Each node $v$ sets $par(v):= R_v(s(v))$, $R_v$ being the routing table of $v$, and sends a message to $par(v)$ to indicate this choice. If $v$ receives such a message from another node $u$, it registers $u$ as its child in the SPF. 
\end{itemize}
\end{mdframed}
Step \ref{itm:SPF1} in \ref{alg:spf} requires no communication since each node already knows the shortest path distances to all other nodes, including the Steiner Terminals, meaning it can be executed locally. Each node $v$ choosing $par(v)$ in step \ref{itm:SPF2} can also be done locally using routing table information, and thus step \ref{itm:SPF2} requires $1$ round of communication of $n-|Z|$ classical messages, since all non-Steiner nodes send one message. 
\begin{claim}
After executing the \ref{alg:spf} procedure, the trees\\
$T_{z_k} = (V_{z_k}, E_{z_k}, W)$ with $V_{z_k} := \{v\in V: s(v) = z_k\}$ and $E_{z_k} := \{v, par(v)\}: v\in V_{z_k}\}$ form an SPF.
\end{claim}
\begin{proof}
i) holds since each Steiner Terminal is closest to itself. iii) is immediate. To see that ii) holds, note that for $v \in V_{z_k}$, $par(v) \in V_{z_k}$ and $\{v, par(v)\}\in E_{z_k}$ as well. Then $par(par(\dots par(v) \dots)) = z_k$ and the entire path to $z_k$ lies in $T_{z_k}$. 
\end{proof}
Hence, after this procedure, we have a distributed SPF across our graph, where each node knows its label, parent, and children of the tree it is in. 

\subsection{Weight Modified MST and Pruning}\label{Weight mod MST}
Finally, we introduce a modification of the edge weights before constructing an MST on that new graph that will be pruned into an approximate Steiner Tree.
These remaining steps stem from a centralized algorithm first proposed by \cite{Kou1981_fast_algo_for_ST_} whose steps can be implemented efficiently in the distributed setting, as in \cite{Saikia_Karmakar2019_SteinerTree_CONGESTCLIQUE}.
We first modify the edge weights as follows:

Partition the edges $E$ into three sets -- \textit{tree edges} $E_F$ as in \ref{def: SPF} that are part of the edge set of the SPF, \textit{intra-tree edges} $E_{IT}$ that are incident on two nodes in the same tree $T_i$ of the SPF, and \textit{inter-tree edges} $E_{XT}$ that are incident on two nodes in different trees of the SPF.
Having each node know which of these its edges belong to can be done in one round by having each node send its neighbors the ID of the terminal it chose as the root of the tree in the SPF that is a part of.
Then the edge weights are modified as follows, denoting the modified weights as $W'$:
\begin{itemize}
    \item[(i):] For $e =(u,v)\in E_T, W'(u,v) := 0$
    \item[(ii):] For $e = (u,v) \in E_{IT}, W'(u,v) := \infty$
    \item[(iii):] For $e = (u,v)\in E_{XT}, W'(u,v) := d(u, Z_u) + W(u,v) + d(v, Z_v)$,
\end{itemize}
noting that $d_G(u, s(u))$ is the shortest-path distance in $G$ from $u$ to its closest Steiner Terminal.

Next, we find a minimum spanning tree on the graph $G' = (V, E, W')$, for which we may implement the classical $\mathcal{O}(1)$ round algorithm proposed by \cite{MST_in_O1_CCM_Nowicki2019}.
On a high level, this constant-round complexity is achieved by sparsification techniques, reducing MST instances to sparse ones, and then solving those efficiently.
We skip the details here and refer the interested reader to \cite{MST_in_O1_CCM_Nowicki2019}. 
%Other efficient algorithms to compute an MST in the congested clique include the $O(\log\log n)$ and $O(\log\log\log n)$ algorithms by \cite{Lotker2005_MST_loglogn}, and \cite{hegeman2015_MST_log3n}, respectively.
After this step, each node knows which of its edges are part of this weight-modified MST, as well as the parent-child relationships in the tree for those edges.

Finally, we prune this MST by removing non-terminal leaf nodes and the corresponding edges.
This is done by each node $v$ sending the ID of its parent in the MST 
%(with respect to the fixed root $r$)
to every other node in the graph.
As a result, each node can locally compute the entire MST and then decide whether or not it connects two Steiner Terminals.
If it does, it decides it is part of the Steiner Tree; otherwise, it broadcasts that it is to be pruned.
Each node that has not been pruned then registers the edges connecting it to non-pruned neighbors as part of the Steiner Tree.
This pruning step takes $2$ rounds and up to $n^2+n$ classical messages.

\subsection{Overall Complexity and Correctness}\label{sec: Steiner alg complexity and correctness} 
In algorithm \ref{sec: steiner algorithm overview}, after step \ref{itm:apsp}, steps \ref{itm:spf}, and \ref{itm:weights} can each be done within $2$ rounds. Walking through \cite{MST_in_O1_CCM_Nowicki2019} reveals that the MST for step \ref{itm:mst} can be found in $54$ rounds, with an additional $2$ rounds sufficing for the pruning. 
Hence, the overall complexity remains dominated by Eq.~\eqref{eq:quantum apsp complexity}. 
Hence, the round complexity is $\tilde{\mathcal{O}}(n^{1/4})$, which is faster than any known classical CONGEST-CLIQUE algorithm to produce an approximate Steiner tree of the same approximation ratio. 
However, as a consequence of the full complexity obtained in \S \ref{sec:complexity}, the regime of $n$ in which this algorithm beats the trivial strategy of sending all information to a single node is also $n>10^{18}$. 
For the same reason, the classical algorithm provided in \cite{Saikia_Karmakar2019_SteinerTree_CONGESTCLIQUE} making use of the APSP subroutine from \cite{CensorHillel2016_AlgebraicMethodsCCM_fast_APSP} discussed in \S \ref{sec: classical apsp complexity} has its complexity mostly characterized by Eq.~\eqref{eq:classical apsp complexity}, so that the regime in which it provides an advantage over the trivial strategy lies in $n>10^{11}$. 
Our algorithm's correctness follows from the correctness of each step together with the correctness of the algorithm by \cite{Kou1981_fast_algo_for_ST_} that implements these steps in a classical, centralized manner.

\begin{section}{Directed Minimum Spanning Tree Algorithm}\label{sec: directed MST}
This section will be concerned with establishing Theorem \ref{thm: DMSTalg} for the Directed Minimum Spanning Tree (DMST) problem, in definition \ref{def:DMST}.
%Even though we can get a quantum speedup for MST in the CONGEST non-clique setting according to \cite{Elkin2014_NoQuantumSpeedups}, and standard MST can be solved in the congested clique in $O(1)$ rounds, the fastest known exact classical algorithm for the DMST problem in the congested clique requires $\tilde{O}(n^{1/3})$ rounds by \cite{fischer2021_DMST}.
%The primary computational bottleneck for the classical algorithm is single-source shortest paths computation.
%Hence, this problem provides another candidate for which an accelerated quantum distributed APSP leads to a speedup, which we show here.
Like \cite{fischer2021_DMST}, we follow the algorithmic ideas first proposed by \cite{lovasz1985_DMST}, implementing them in the quantum CONGEST-CLIQUE. Specifically, we will use $\log n$ calls to the APSP and routing tables scheme described in \S \ref{sec: APSP and routing}, so that in our case, we retrieve complexity $\tilde{\mathcal{O}}(n^{1/4})$ and success probability $(1-\frac{1}{poly(n)})^{\log n} = 1-\frac{1}{poly(n)}$.

%Before describing the algorithm, we note that the accelerated DMST algorithm we obtain also leverages the efficient APSP and routing table computation described in Section \ref{sec: APSP and routing}.
%In fact, \cite[Theorem 1]{fischer2021_DMST} explicitly characterizes the algorithm's complexity to be the same as in deterministic single-source shortest paths, up to logarithmic factors.
%On a high level, we use the quantum-accelerated APSP with routing tables instead of SSSP. However, our APSP algorithm is not deterministic but has some error rate, so the success probability of putting together this APSP scheme with the algorithm of \cite{fischer2021_DMST} needs to be verified.
%We proceed by giving an overview of the algorithm in question, presenting that it can be implemented using $O(\log n)$ calls to the APSP scheme of Section \ref{sec: APSP and routing}, which will recover a $(1-\frac{1}{poly(n)})^{\log n}$  success probability.

Before describing the algorithm, we need to establish some preliminaries and terminology for the procedures executed during the algorithm, especially the ideas of shrinking vertices into \textit{super-vertices} and tracking a set $H$ of specific edges as first described in \cite{edmonds1967_DMST}. We use the following language to discuss super-vertices and related objects.
\begin{definition}\label{super-vertex}
A {\em super-vertex set} $\V^*:= \{V^*_1, \dots, V^*_t\}$ for a graph $G = (V, E, W)$ is a partition of $V$, and each $V^*_i$ is called a {\em super-vertex}. We will call a super-vertex {\em simple} if $V^*$ is a singleton. 
The corresponding {\em minor} $G^*:= (\V^*, E^*, W^*)$ is the graph obtained by creating edges $(V^*_i, V^*_j)$ with weight $W^*(V^*_i, V^*_j):= \min \{W(v_i,v_j): v_i\in V^*_i, v_j \in V^*_j\}$.
\end{definition}

%It should be noted that the super-graph being a multigraph will not be relevant in this paper, since we will only contract cycles of subgraphs that have one incoming edge per node. 
%Hence, all non-simple super-vertices will have at most one outgoing edge, and no incoming edge in the super-graph, and it is obvious that all simple super-vertices are incident on at most one outgoing and one incoming edge.
Notably, we continue to follow the convention of an edge of weight $\infty$ being equivalent to not having an edge. We will refer to creating a super-vertex $V^*$ as {\em contracting} the vertices in $V^*$ into a super-vertex. 
\subsection{Edmonds' Centralized DMST Algorithm}
We provide a brief overview of the algorithm proposed in \cite{edmonds1967_DMST}, which presents the core ideas of the super-vertex-based approach.
The following algorithm produces a DMST for $G$:
\begin{mdframed}\namedlabel{alg:edmondsDMST}{\texttt{Edmonds DMST Algorithm}}
% \begin{itemize}
\item[] Input: An integer-weighted digraph and a root node $r$. 
\item[] Output: A DMST for $G$ rooted at $r$. 
% \end{itemize}
\begin{enumerate}
    \item[\namedlabel{edmondsDMST init}{1.}]Initialize a subgraph $H$ with the same vertex set as G by subtracting for each node the minimum incoming edge weight from all its incoming edges, and selecting exactly one incoming zero-weight edge for each non-root  node of $G$. Set $G_0 = G, H_0 = H, t = 0$.
    \item [\namedlabel{edmondsDMST while}{2.}]WHILE $H_t$ is not a tree:
    \begin{enumerate}
        \item For each cycle of $H$, contract the nodes on that cycle into a super-vertex. Consider all non-contracted nodes as simple super-vertices, and obtain a new graph $G_{t+1}$ as the resulting minor. 
        \item If there is a non-root node of $G_{t+1}$ with no incoming edges, report a failure. Otherwise, obtain a subgraph $H_{t+1}$ by, for each non-root node of $G_{t+1}$, subtracting the minimum incoming edge weight from all its incoming edges, and selecting exactly one incoming zero-weight edge for each non-root, updating $t\leftarrow t+1$. 
    \end{enumerate}
    \item [\namedlabel{edmondsDMST for}{3.}] Let $B_t = H_t$. FOR $k \in (t, t-1, \dots, 1)$:
    \begin{enumerate}
        \item Obtain $B_{k-1}'$ by expanding the non-simple super-vertices of $B_k$ and selecting all but one of the edges for each of the previously contracted cycles of $H_k$ to add to $B_{k-1}$.
    \end{enumerate}
        
    \item [\namedlabel{edmondsDMST return}{4.}]Return $B_0$.
    
\end{enumerate}
\end{mdframed}
Note that the edge weight modifications modify the weight of all directed spanning trees equally, so optimality is unaffected.
In step \ref{edmondsDMST while}, if $H_t$ is a tree, it is an optimal DMST for the current graph $G_t$.
Otherwise, it contains at least one directed cycle, so that indeed step \ref{edmondsDMST while} is valid. Hence, at the beginning of step \ref{edmondsDMST for}, $B_t$ is a DMST for $G_t$. Then the first iteration produces $B_{t-1}$ a DMST for $G_{t-1}$ since only edges of zero weight were added, and $B_{t-1}$ will have no cycles. The same holds for $B_{t-2}, B_{t-3}, \dots, B_0$, for which $B_0$ corresponds to the DMST for the original graph $G$. If the algorithm reports a failure at some point, no spanning tree rooted at $r$ exists for the graph, since a failure is reported only when there is an isolated non-root connected component in $G_{t+1}$.

Note that in iteration $t$ of step \ref{edmondsDMST while}, $H$ has one cycle for each of its connected components that does not contain the root node.
Hence, the drawback of this algorithm is that we may apply up to $\mathcal{O}(n)$ steps of shrinking cycles. This shortcoming is remedied by a more efficient method of selecting how to shrink nodes into super-vertices in \cite{lovasz1985_DMST}, such that only $\log n$ shrinking cycle steps take place.
%, described in the following subsection.

\subsection{Lovasz' Shrinking Iterations}
We devote this subsection to discuss the shrinking step of \cite{lovasz1985_DMST} that will be repeated $\log n$ times in place of step \ref{edmondsDMST while} of Edmonds' algorithm to obtain Lovasz' DMST algorithm. \newpage
\begin{mdframed}
{\texttt{Lovasz' Shrinking Iteration}}
\namedlabel{alg:LSI}{\texttt{LSI}}\\
Input: A directed, weighted graph $G = (V, E, W)$ and a root node $r\in V$. \\
Output: Either a new graph $G^*$, or a success flag and a DMST $H$ of $G$. 
\begin{itemize}

\item [\namedlabel{itm:lovasz1}{1.}] If there is a non-root node of $G$ with no incoming edges, report a failure. Otherwise, for each non-root node of $G$, subtract the minimum incoming edge weight from all its incoming edges.
Select exactly one incoming zero-weight edge for each non-root node to create a subgraph $H$ of $G$ with those edges. 

\item[\namedlabel{itm:lovasz2}{2}] Find all cycles of $H$, and denote them $H_1, \dots, H_C.$ If $H$ has no cycles, abort the iteration and return (SUCCESS, H). For $j=1, \dots, C$, find the set $V_j$ of nodes that dipaths in $H$ from $H_j$ can reach.

\item[\namedlabel{itm:lovasz3}{3.}] Compute the All-Pairs-Shortest-Path distances in $G$.
%$d(v_i, v_j)$ for all $(v_i, v_j) \in V \times V$.

\item[\namedlabel{itm:lovasz4}{4.}] For each node $v\in V$, denote $d_j(v):= \min \{d(v, u): u\in H_j\}$. For each $j = 1, \dots, C$, set $\beta_j := \min \{d_j(v): v\in V(G)\setminus \V_j\}$ and $U_j := \{u\in V_j: d_j(u)\leq \beta_j\}$.

\item[\namedlabel{itm:lovasz5}{5.}] Create a minor $G^*$ by contracting each $U_j$ into a super-vertex $U_j^*$, considering all other vertices of $G$ as simple super-vertices $V^*_1, \dots, V^*_k$. For each vertex $N^*$ of $G^*$, let the edge weights in $G^*$ be: 
\begin{align*}
    W^*_{N^* U^*_j} &= \min\{W_{vu}: v\in N^*, u\in U^*_j\} - \beta_j + \min\{d_j(u): u\in U^*_j\}  \\
    &\text{for all } j= 1, \dots, C, \text{ and} \\
W^*_{N^* V^*} &= \min\{W_{v V^*}: v\in N^*\} \\
&\text{for all the simple super-vertices } V^* \text{ of } G^*. 
\end{align*}

\item[6.] Return $G^*$.
\end{itemize}
\end{mdframed}

To summarize these iterations: The minimum-weight incoming edge of each node is selected. 
That weight is subtracted from the weights of every incoming edge to that node, and one of those edges with new weight 0 is selected for each node to create a subgraph $H$. If $H$ is a tree, we are done. Otherwise, we find all cycles of the resulting directed subgraph, then compute APSP and determine the $V_j, U_j,$ and $\beta_j$, which we use to define a new graph with some nodes of the original $G$ contracted into super-vertices. 

The main result for the DMST problem in \cite{lovasz1985_DMST} is that replacing (a) and (b) of step \ref{edmondsDMST while} in the \ref{alg:edmondsDMST}, taking the new $H$ obtained at each iteration to be $H_{t+1}$ and the $G^*$ to be $G_{t+1}$, leads to no more than $\ceil{\log n}$ such shrinking iterations needed before a success is reported.

\subsubsection{Quantum Distributed Implementation}
Our goal is to implement the Lovasz iterations in the quantum distributed setting in $\tilde{\mathcal{O}}(n^{1/4})$ rounds by making use of quantum APSP of \S \ref{sec: APSP and routing}.
In the distributed setting, processor nodes cannot directly be shrunk into super-vertices.
As in \cite{fischer2021_DMST}, we reconcile this issue by representing the super-vertex contractions within the nodes through \textit{soft contractions}.

First, note that a convenient way to track what nodes we want to consider merging into a super-vertex is to keep a mapping $sID: V \rightarrow S$, where $S$ is a set of super-vertex IDs, which we can just take to be the IDs of the original nodes.
We will refer to a pair of $(G, sID)$ as an \textit{annotated graph}. An annotated graph naturally corresponds to some \emph{minor} of $G$, namely, the minor obtained by contracting all vertices sharing a super-vertex ID into a super-vertex. 

\begin{definition}[Soft Contractions]
For an annotated graph $(G, sID)$, a set of active edges $H$, and active component $H_i$ with corresponding weight modifiers $\beta_i$, and a subset $A \subset S$ of super-vertices, the \textit{soft contraction} of $H_i$ in G is the annotated graph $(G^{H_i}, sID')$ obtained by taking $G^{H_i} = (V, E, W')$ with 
\begin{itemize}
    \item $W'_{uv} = 0 $ if $sID(u) = sID(v)$
    \item $W'_{uv} = W_{uv} + dist_{G(A)}(v, C(H_i))-\beta_i$ if $u \in V \setminus A$ and $v \in A$
    \item $W'_{uv} = W_{uv}$ otherwise
\end{itemize}
and updating the mapping $sID$ to $sID'$ defined by $sID'(v) = sID(v),  \forall v \notin A$, $sID'(v) = \min\{sID(u): u\in A\}$.
%Note that it need not be the $\min sID$ chosen but that any node in $A$ works as long as the nodes can coordinate to select the same one, so choosing the $\min sID$ is convenient.
\end{definition} 

\subsubsection{Quantum Distributed Lovasz' Iteration}
We provide here a quantum distributed implementation of Lovasz' iteration that we will form the core of our DMST algorithm.
%Each iteration consists of the following steps, and we need at most $2 \log n$ such iterations as with Lovasz's iterations.
\begin{mdframed}
{\texttt{Quantum Distributed Lovasz' Iteration}}
\namedlabel{alg:QDLSI}{{\texttt{QDLSI}}} \\
Input: A directed, weighted, graph $G = (V, E, W)$ with annotations $sID$ and a subgraph $H$. \\
Output: A new graph $G^*$ with annotations $sID'$, or a success flag and a DMST $H$ of $G$.
\begin{itemize} 
    \item[\namedlabel{itm:QLI1}{1}:] Have all nodes learn all edges of $H$, as well as the current super-vertices. 
    \item[\namedlabel{itm:QLI2}{2}:] For each connected component $H_i \subset H$, denote by $C(H_i)$ the cycle of $H_i$. Let $c(H_i)$ be the node with maximal ID in $C(H_i)$, which each node can locally compute. 
    %note: the max ID is fine here, and we can skip the discussion about the center choice of sec 5.2 Fischer and Oshman since we are only in the congested clique (the center choice matters for congest only)
    \item[\namedlabel{itm:QLI3}{3}:] Run the quantum algorithm for APSP and routing tables described in \S \ref{sec: APSP and routing} on this graph, or report a failure if it fails.
    \item[\namedlabel{itm:QLI4}{4}:] For each $i$, determine an edge $v_iu_i$, $v_i \notin H_i, u_i \in H_i$ minimizing $\beta_i := W_{v_iu_i} + d_{G}(u_i, c(H_i))$, and broadcast both to all nodes in $H_i$. 
    \item[\namedlabel{itm:QLI5}{5}:] Each node $v_i$ in each $H_i$ applies the following updates $locally$:
    \begin{itemize}
        \item Soft-contract $H_i$ at level $\beta_i$ to soft-contract all super-vertices with distance $\beta_i$ to $C(H_i)$ into one super-vertex, with each contracted node updating its super-vertex ID to $c(H_i)$
        \item add edge $v_i u_i$ to $H$, effectively merging $H_i$ with another active component of $H$
        %\item for the node $c(H_i)$, update $W(c(H_i)) \leftarrow W(c(H_i)) + \beta_i$
    \end{itemize}
\end{itemize}
\end{mdframed}

We can follow exactly the steps of Lovasz's DMST algorithm, distributedly by replacing steps \ref{itm:QLI2}-\ref{itm:QLI5} of the \ref{alg:LSI} with this quantum-distributed version. The following ensues: 
\begin{lemma}\label{lemma:DMSTalg}
If none of the APSP and routing table subroutines fail, within $\ceil{\log n}$ iterations of the \ref{alg:QDLSI}, $H$ is a single connected component.
\end{lemma}

\begin{lemma}\label{lemma:APSP prob in DMST}
With probability $(1-\frac{1}{poly(n)})^{\log n}$, all the APSP and routing table subroutines in step \ref{itm:QLI3} succeed.
\end{lemma}

Lemmas \ref{lemma:DMSTalg} and \ref{lemma:APSP prob in DMST} then together imply Theorem \ref{thm: DMSTalg}.
%, noting $(1-\frac{1}{poly(n)})^{\log n} = 1-\frac{1}{poly(n)}$. 
Within $\ceil{\log n}$ iterations, only one active component remains: the root component. This active component can then be expanded to a full DMST on $G$ within $\ceil{\log n}$ rounds, as detailed in \cite[\S 7]{fischer2021_DMST} or the \ref{alg:unpacking} procedure in \S \ref{sec:expanding the dmst} of the appendix.
All messages in the algorithm other than those for computing the APSP in \ref{alg:QDLSI} may be classical.
We provide here the full algorithm for completeness:
\begin{mdframed}
\namedlabel{alg:quantumDMST}{\texttt{Quantum DMST Algorithm}}
% \begin{itemize}
\item[] Input: An integer-weighted digraph and a root node $r$. 
\item[] Output: A DMST for $G$ rooted at $r$. 
% \end{itemize}
\begin{enumerate}
    \item[\namedlabel{quantumDMST init}{1.}]Initialize a subgraph $H$ with the same vertex set as G by subtracting for each node the minimum incoming edge weight from all its incoming edges, and selecting exactly one incoming zero-weight edge for each non-root  node of $G$. Set $t = 0, H_0 = H$, and $G_0 = G$ with annotations $sID_0$ to be the identity mapping.
    \item [\namedlabel{quantumDMST while}{2.}]WHILE: $H_t$ is not a single component
    \begin{enumerate}
        \item Run \ref{alg:QDLSI} with inputs $H_t$, $(G_t, sID_t)$ to obtain $H_{t+1}$, $(G_{t+1}, sID{t+1})$ as outputs. Increment $t \leftarrow t+1$.
    \end{enumerate}
    \item [\namedlabel{quantumDMST unpack}{3.}] Let $T_t := H_t$. For $k=t, \dots, 1$: For each super-vertex of the $k^{th}$ iteration of \ref{alg:QDLSI} applied, simultaneously run the \ref{alg:unpacking} procedure with input tree $T_k$ to obtain $T_{k-1}$.
    \item [\namedlabel{quantumDMST return}{4.}] Return $T_0$ as the distributed minimum spanning tree. 
\end{enumerate}
\end{mdframed}
%This expansion can be done because, in step \ref{itm:QLI1}, the nodes learn all edges of the $H$ of the current iteration and can store this information for the remainder of the algorithm.
%Hence, the algorithm also requires the nodes to have $O(n\cdot \log^2(n)\log(W))$ memory since each $H$ has $n$ edges consisting of two node IDs of length $\log n$ and edge weights of length $\log W$, and there are up to $\log n$ such $H$'s. 

\subsection{Complexity}\label{sec:dmst complexity}
In the \ref{alg:QDLSI}, all steps other than the APSP step \ref{itm:QLI3} of the quantum Lovasz iteration can be implemented within $2$ rounds. In particular, to have all nodes know some tree on G for which each node knows its parent, every node can simply broadcast its parent edge and weight. Since this iteration is used up to $\ceil{\log(n)}$ times and expanding the DMST at the end of the algorithm also takes logarithmically many rounds, we obtain a complexity dominated by the APSP computation of $\tilde{\mathcal{O}}(n^{1/4})$, a better asymptotic rate than any known classical CONGEST-CLIQUE algorithm. However, beyond the $\tilde{\mathcal{O}}$, the complexity is largely characterized by $\log(n) \cdot f(n)$, with $f(n)$ as in Eq.~\eqref{eq:quantum apsp complexity}. In order to have $\log(n) f(n) < n$ to improve upon the trivial strategy of having a single node solve the problem, we then need $n > 10^{21}$. Using the classical APSP from \cite{CensorHillel2016_AlgebraicMethodsCCM_fast_APSP} in place of the quantum APSP of \S \ref{sec: APSP and routing} as done in \cite{fischer2021_DMST} to attain the $\tilde{\mathcal{O}}(n^{1/3})$ complexity in the cCCM, one would need $\log(n)\cdot g(n) < n$ to beat the trivial strategy, with $g$ as in Eq.~\eqref{eq:classical apsp complexity}, or more than $n> 10^{14}$. 

\section{Discussion and Future Work}
We have provided algorithms in the Quantum CONGEST-CLIQUE model for computing approximately optimal Steiner Trees and exact Directed Minimum Spanning trees that use asymptotically fewer rounds than their classical known counterparts.
As Steiner Tree and Minimum Spanning Trees cannot benefit from quantum communication in the CONGEST (non-clique) model, the algorithms reveal how quantum communication can be exploited thanks to  the CONGEST-CLIQUE setting.
A few open questions remain as well.
In particular, there exist many generalizations of the Steiner Tree problem, so these may be a natural starting point to attempt to generalize the results.
A helpful overview of Steiner-type problems can be found in \cite{SteinerWebsite}. 
Regarding the DMST, it may be difficult to generalize a similar approach to closely related problems. 
Since the standard MST can be solved in a (relatively small) constant number of rounds in the classical CONGEST-CLIQUE, no significant quantum speedup is possible. Other interesting MST-type problems are the bounded-degree and minimum-degree spanning tree problems. However, even the bounded-degree decision problem on an unweighted graph, ``does $G$ have a spanning tree of degree at most $k$?'' is NP-complete, unlike the DMST, so we suspect that other techniques would need to be employed. \cite{dinitz_MinDegreeMST} provides a classical distributed approximation algorithm for the problem.
Additionally, we have traced many constants and $\log$ factors throughout our description of the above algorithms, which, as shown, would need to be significantly improved for these and related algorithms to be practical. 
Hence, a natural avenue for future work is to work towards such practical improvements. 
Beyond the scope of the particular algorithms involved, we hope to help the community recognize the severity with which the practicality of algorithms is affected by logarithmic factors that may be obscured by $\tilde{\mathcal{O}}$ notation, and thus encourage fellow researchers to present the full complexity of their algorithms beyond asymptotics. Particularly in a model like CONGEST-CLIQUE, where problems can always be solved trivially in $n$ rounds, these logarithmic factors should clearly not be taken lightly. 
Further, a question of potential practical interest would be to ask the following:
What algorithms solving the discussed problems are the most efficient with respect to rounds needed in the CONGEST-CLIQUE in the regimes of $n$ in which the discussed algorithms are impractical? %In light of the $n>10^{20}$ required to have an actual advantage in the number of rounds needed over the trivial strategy, but our algorithms beating any known classical algorithms asymptotically, we interpret our results as primarily of theoretical interest unless the $\log$ factors involved in the complexity that are the root of this impracticality can be improved. 

\end{section}

\section*{Acknowledgements}
We are grateful for support from the NASA Ames Research Center, from the NASA SCaN program, and from DARPA under IAA 8839, Annex 130. PK and DB acknowledge support from the NASA Academic Mission Services (contract NNA16BD14C).
%P.K. acknowledges support from the NASA/USRA Feynman Quantum Academy Internship program.
The authors thank Ojas Parekh for helpful input and discussions regarding the arborescence problem, Shon Grabbe for ongoing useful discussions, and Filip Maciejewskifor for helpful feedback on the work.

\newpage

{\small
  \bibliographystyle{apacite} %uses apastyle package instead of natbib
 \bibliography{quantumcite}
}

\newpage 

\section{Appendix}

\subsection{Proof of claim \ref{claim: witness from dist}}\label{app:witnessfromdist}

For an $n \cross n$ integer matrix $W$, obtain matrices $W'$ and $W''$ by taking $W'_{ij} = n W_{ij}+ j -1$ and $W''_{ji} = nW_{ji}$. Set $D = W'\star W''$. We aim to show that
$\floor*{ {\dfrac{D}{n}} } = W^{2, \star}$ and 
$(D \mod n)+1$ is a witness matrix for $W^{2, \star}.$

\begin{proof}
 \begin{itemize}
    \item[] 
     \item[(i)] We have 
     \begin{align*}
         \floor* { {\dfrac{D}{n}} } _{ij} &= 
         \floor*{\min_{k\in [n]}\left  \{nW_{ik} +k -1 + nW_{kj}  \right \} /n}
         = \floor*{ \min_{k\in [n]} \left\{W_{ik} +W_{kj} + \frac{k-1}{n} \right \}}\\
         &= W^2_{ij} + \floor*{\min_{k\in [n]} \left  \{\frac{k-1}{n}: W_{ik} +W_{kj} = W^2_{ij} \right \}}
         = W^2_{ij}.
     \end{align*}

     \item[(ii)]Next, 
     \begin{align*}
         D_{ij} = nW^2_{ij} + \floor*{\min_{k\in [n]} \left  \{k-1: W_{ik} +W_{kj} = W^2_{ij} \right \}}
     \end{align*}
     gives us 
     $$
     (D \mod n)+ 1 = \floor*{\min_{k\in [n]} \left  \{k-1: W_{ik} +W_{kj} = W^2_{ij} \right \}} +1 = \min_{k\in [n]} \left  \{k: W_{ik} +W_{kj} = W^2_{ij} \right \},
     $$
     which proves the claim. 
 \end{itemize}
\end{proof}

\subsection{The \texorpdfstring{$\alpha$}{alpha} > 0 case}\label{app:alpha nonzero case}
The strategy will be to assign each $v_{(i,j,k)}\in \mathbb{V}$ into classes in accordance with approximately how many negative triangles are in $U_i\times U_j \times U'_k$ before starting the search. 

To assign each node to a class, we use the routine \ref{alg:IdentifyClass} of \cite{IzumiLeGallMag2019_APSP_QuantumDist}, also described in the main text.

The main body of this paper discussed the special case assuming $\alpha = 0$. Hence we now consider the $\alpha > 0$ case.

For each $\alpha \in \mathbb{N}$, let us denote $c_{i,j,k}$ the smallest nonnegative integer satisfying $d_{i,j,k} < 10\cdot 2^c\log n$, and 
\begin{align}
    V_\alpha &:= \{v_{(i,j,k)}: c_{i,j,k} = \alpha\}\\
    V_\alpha[i,j] &:= \{U_k' \in \mathcal{U}': v_{(i,j,k)}\in V_\alpha\} \label{:def Valpha}
\end{align} 
for any $i, j \in [n^{1/4}]$. Notably, $\mathcal{P}(i,j)$ contains at most $\sqrt{n}$ edges, so that $d_{i,j,k}\leq \sqrt{n}$ as well. 
Hence, $c=\frac{1}{2}\log n$ provides an upper bound for the minimum in step \ref{itm:IC2}.
The important immediate consequence is that we only need to consider $V_\alpha$ up to at most $\alpha = \frac{1}{2}\log n$. 

%I'm deciding to skip some things here, such as proposition 5, and lemmas 3 and 4. 
%We can discuss whether we want those details or leave them to the reference and just mention the parts of them that we need to describe the algorithm. 

\begin{lemma}\label{lemma:IClemma appendix}
The \ref{alg:IdentifyClass} algorithm and the resulting $V_\alpha$ satisfy the following statements with probability at least $1-2/n$: 
\begin{itemize}
    \item [\namedlabel{itm:IClemma appendixi}{(i)}:]  The algorithm does not abort
    \item [\namedlabel{itm:IClemma appendixii}{(ii)}:] $|\Delta(i,j,k)| \leq 2n$
    \item [\namedlabel{itm:IClemma appendixiii}{(iii)}:] For $\alpha >0$, $v_{(i,j,k)} \in V_\alpha$, we have $2^{\alpha-3}n\leq|\Delta(i,j,k)|\leq2^{\alpha+1}n$. 
    \item [\namedlabel{itm:IClemma appendixiv}{(iv)}:] $|\Lambda_x(i,j)\cap \Delta(i,j,k)|\leq 100\cdot 2^\alpha \sqrt{n}\log n$ for $i,j\in [n^{1/4}]$ and $\alpha \in \N$. 
\end{itemize}
\end{lemma}
This provides an adapted version of lemma \ref{lemma:IClemma} for the $\alpha >0$ case.

The following lemma provides a tool that will allow for "duplication" of information to avoid message congestion in the network in the \ref{alg:EvaluationB} procedure. 

\begin{lemma}\label{lemma:valpha_sizebound}
For all $\alpha \geq 0$ and $i,j \in [n^{1/4}]$,
\begin{align}
    |V_\alpha[i,j]| \leq \dfrac{720\sqrt{n}\log n}{2^\alpha}
\end{align}
\end{lemma}

\begin{proof}
The $\alpha = 0$ case is immediate since $|\mathcal{U}'|  = \sqrt{n}$, so consider $\alpha \geq 1$. The ``promise'' in the FEWP subroutine we are in guarantees that for all $(u,v)\in S, \Gamma(u,v) \leq 90 \log n$, so that for any $i,j\in [n^{1/4}]$, each edge in $\mathcal{P}(U_i, U_j)\cap S$ has at most $90\log n$ other nodes forming a negative triangle with it, leading to the inequality 
$$
\sum_{k: v_{(i,j,k)}\in V_\alpha}|\Delta(i,j,k)| \leq 90n^{3/2}\log n. 
$$
Using $|\Delta(i,j,k)| \geq 2^{\alpha-3}n$ from part \ref{itm:IClemma appendixi} of lemma \ref{lemma:IClemma appendix}, the conclusion follows. 
\end{proof}

We now describe the implementation of step \ref{itm:CP3} of the \ref{alg:computepairs} procedure for the $\alpha > 0$ case. 
\begin{mdframed}
\begin{itemize}
    \item [\namedlabel{itm:CP3.1 appendix}{3.1}:] Each node executes the \ref{alg:IdentifyClass} procedure. 
    \item [\namedlabel{itm:CP3.2 appendix}{3.2}:] For each $\alpha$:\\
    For every $l\in [m]$, every node $v_{(i,j,k)}$ executes a quantum search to find whether there is a $U_k'\in V_\alpha[U_i, U_j]$ with some $w\in U_k'$ forming a negative triangle $(u^k_l, v^k_l, w)$ in $G$, and then reports all the pairs $u^k_l v^k_l$ for which such a $U_k'$ was found.
\end{itemize}
\end{mdframed}

The $\alpha =0$ case was described in the main text. We proceed to describe the classical procedure for invoking theorem \ref{thm: distquantsearches} to obtain the speedup for the general $\alpha$ case, as in \cite[\S 5.3.2]{IzumiLeGallMag2019_APSP_QuantumDist}. Some technical precautions must be taken to avoid congestion of messages between nodes. This crucially relies on information duplication to effectively increase bandwidth between nodes. Lemma \ref{lemma:valpha_sizebound} provides a strong bound for the size of each $V_\alpha$. For this duplication of the information stored by the relevant nodes, a new labeling scheme is convenient. Suppose for simplicity that $C_\alpha := 2^\alpha/(720\log n)$ is an integer, and assign each node a label $({\ve u}, \bv, \w, y)\in V_\alpha \times [C_\alpha]$, which is possible due to the bound of lemma \ref{lemma:valpha_sizebound}. 
The following \ref{alg:EvaluationB} implementable in  $\mathcal{O}(\log n)$ rounds (using a slightly sharper complexity analysis than \cite{IzumiLeGallMag2019_APSP_QuantumDist}) can then be used for invoking theorem \ref{thm: distquantsearches}:

\begin{mdframed}\namedlabel{alg:EvaluationB}{\texttt{EvaluationB}}\\
Input: A list $(\w_1^k, \dots, \w_m^k)$ of elements of $V_\alpha[{\ve u}, \bv]$ assigned to each node $k= ({\ve u}, \bv, x)$.\\
Promise: $|L^k_\w|\leq 800\cdot 2^\alpha \sqrt{n}\log n$ for each node $k$ and all $|w \in V_\alpha[{\ve u}, \bv]$. \\
Output: Every node $k = ({\ve u}, \bv, x)$ outputs for each $\ell \in [m]$ whether some $w\in \w^k_l$ forms a negative triangle $\{u^k_\ell, v_\ell^k, w\}$. 

\begin{itemize}
    \item[\namedlabel{itm:eval0}{0}.] Every node $({\ve u}, \bv, \w) \in V_\alpha$ broadcasts the edge information loaded in step \ref{itm:CP1} of \ref{alg:computepairs} to $({\ve u}, \bv, \w, y)$ for each $y\in [C_\alpha]$. 
    \item[\namedlabel{itm:eval1}{1}.]  Every node $({\ve u}, \bv, x)$ splits each $L_\w^k$ into smaller sublists $L_{\w, 1}^k, \dots, L^k_{\w, C_\alpha}$ for each $\w$, with each sublist containing up to $\ceil{|L_\w^k|/C_\alpha} = \ceil{800\cdot 720\sqrt{n}\log^2 n}$ elements,and sends each $L_{\w, y}^k$ to node $({\ve u}, \bv, \w, y)$ along with the relevant edge weights.
    \item[\namedlabel{itm:eval2}{2}.]  Every $({\ve u}, \bv, \w, y)$ node returns the truth value
    $$
    \min_{w\in \w}\{W_{uw}+W_{wv}\} \leq W_{vu}
    $$ to node $k$ for each $uv \in L_{\w, y}^k$ received in step \ref{itm:eval1}.
\end{itemize}
\end{mdframed}
For each value of $\alpha$, we separately solve step \ref{itm:CP3.2 appendix} of the \ref{alg:computepairs}  procedure. Since lemma \ref{lemma:valpha_sizebound} tells us that there are $C_\alpha$ times more nodes not in $V_\alpha$ than there are in $V_\alpha$, every node in $V_\alpha$ can use $C_\alpha$ of those nodes not in $V_\alpha$ to relay messages and effectively increase its message bandwidth, which is exactly what \ref{alg:EvaluationB} takes advantage of. Steps \ref{itm:eval1} and \ref{itm:eval2} of the procedure take up to $2\cdot\ceil{|L_\w^k|}/n \leq 1600\cdot\log n$ rounds, since lists of size $\ceil{|L_\w^k|/C_\alpha}$ are sent to $C_\alpha$ nodes, and the bound on $\alpha$ gives $\ceil{|L_\w^k|} \leq 800n\log(n)$.

\subsubsection{Complexity of the Classical Analogue}\label{sec: appendix: classical apsp complexity details}
This subsection of the appendix serves to provide some supplemental information to \S \ref{sec: classical apsp complexity} discussing the complexity of an algorithm for APSP with routing tables in the CONGEST-CLIQUE as proposed in \cite{CensorHillel2016_AlgebraicMethodsCCM_fast_APSP} that has complexity $\tilde{\mathcal{O}}(n^{1/3})$. Note that \cite[corollary 6]{CensorHillel2016_AlgebraicMethodsCCM_fast_APSP} applied to APSP {\em distance} computations only, whereas the routing table computations are discussed in \cite[\S 3.4]{CensorHillel2016_AlgebraicMethodsCCM_fast_APSP}. As shown there, $\mathcal{O}(\log^3)$ distance products (without witnesses) are needed to compute one distance product with a witness matrix. More precisely:

\begin{enumerate}
    \item Obtaining a witness matrix when witnesses are unique requires $\log(n)$ distance products. 
    \item The procedure for finding witnesses in the general case calls the procedure to find witnesses in the unique witness case $\mathcal{O}(\log^2 n)$ times, or $2\cdot \log^2 n$ times if $c=2$ is deemed as sufficient for the success probability.  
    \item $\log n$ such distance products with witnesses are needed for the APSP algorithm with routing tables.
\end{enumerate}

Then $2\log^4 n$ distance products are computed in total for one distance product with witnesses. The distance product via the semi-ring matrix multiplication algorithm of \cite[\S 2.1]{CensorHillel2016_AlgebraicMethodsCCM_fast_APSP} uses $10n^{1/3}$ rounds ($4n^{1/3}$ for its steps 1 and 2, and $2n^{1/3}$ for step 3) using lemma \ref{lemma: dolev routing lemma}, and hence one obtains the full round complexity of
\begin{align}
  10n^{1/3}\cdot 2\log(n)^4 = g(n).  
\end{align}

\subsection{Expanding the DMST in the Distributed Setting}\label{sec:expanding the dmst}
We handle the expansion of the DMST in the same way as in \cite[\S 7]{fischer2021_DMST}, borrowing much of their discussion for our description here. However, as we have computed APSP distances along the way in place of SSSP, `unpacking' the DMST becomes a bit simpler in our case. \\
    Consider a component $H_i$ in one of the iterations of \ref{alg:QDLSI}, with input graph for the iteration being $G_i$. For each contraction in \ref{alg:QDLSI}, we determined edges $v_iu_i$, $v_i \notin H_i, u_i \in H_i$ minimizing $\beta_i := W_{v_iu_i} + d_{G}(u_i, c(H_i))$ to contract nodes. Recall that what happens in the iteration is that the cycle $c(H_i)$ and all nodes that have distance $\beta_i$ to $c(H_i)$ are contracted into one super-vertex. Denote that super-vertex by $V_{H_i,\beta_i}^*$. Let $G_{i+1}$ denote the graph obtained after this contraction. Then our goal, given a DMST $T_{i+1}$ for $G_{i+1}$, is to recover $G_i$ along with a DMST $T_i$ for $G_i$. We make use of the following \ref{alg:unpacking} operation of \cite[\S 7]{fischer2021_DMST}: 

\begin{mdframed}\namedlabel{alg:unpacking}{{\texttt{Unpacking}}}
% \begin{itemize}
\item[] Input: A digraph $G_{i+1}$ with a DMST $T_{i+1}$ with root $r$, a set of edges $H_i$ as in \ref{alg:quantumDMST}, a node $V^*_{H_i,\beta_i}$ of $G_{I+1}$ marked as a super-vertex, a set $c(H_i)$ of the nodes contracted into it, and $G_i$ the graph before contracting $c(H_i)$. 
\item[] Output: A DMST $T_i$ for $G_i$ rooted at $r$. 
% \end{itemize}
\begin{itemize}
    \item[\namedlabel{itm:unpacking1}{1}:] For any $v_1, v_2 \notin V_{H_i,\beta_i}^*$, let edge $v_1v_2 \in T_{i}$ iff $v_1v_2 \in T_{i+1}$. 
    \item[\namedlabel{itm:unpacking2}{2}:] For $uV_{H_i,\beta_i}^* \in T_{i+1}$, which exists since $T_{i+1}$ is a DMST for $G_{i+1}$, denote the edge \\ $uv^*:= argmin_{uv: v\in V_{H_i,\beta_i}^*, u: \exists uv \in G_{i+1}}W_{vu} + d_{G}(u, c(H_i))$. Add $uv^*$ and the shortest path $\zeta$ connecting $v^*$ to $c(H_i)$ to $T_i$.
    \item[\namedlabel{itm:unpacking3}{3}:] For any edge $V_{H_i,\beta_i}^*u\in T_{i+1}$ outgoing from the contracted super-vertex, add the edge\\ $argmin_{vu: v\in V_{H_i,\beta_i}^*} W_{vu}^{G_i}$ to $T_{i}$. 
    \item[\namedlabel{itm:unpacking4}{4}:] Add all edges $H_i \setminus \delta^{in}(\zeta)$ to $T_i$, where $\delta^{in}(\zeta)$ denotes all edges incoming on $\zeta$.
\end{itemize}
\end{mdframed}
At the end of this procedure, $T_i$ is a DMST for $G_i$ \cite[lemma 8]{fischer2021_DMST}. We now describe how it can be implemented distributedly, needing only classical messages and information. For every contracted super-vertex, the following steps can be implemented at the same time, as will become clear in how the steps are executed for the nodes of each contracted super-vertex. Let us focus on unpacking one super-vertex $V_{H_i,\beta_i}^*$. Each node knows its neighbors in $G_i$, and every node's super-vertex ID in $G_{i}$ and $G_{i+1}$, since each node stores this information before the initial contraction to $G_{i+1}$ in \ref{alg:QDLSI} happens. Hence, step \ref{itm:unpacking1} can be done locally at each node without any communication. Step \ref{itm:unpacking2} can be done by first having each node $v\in V_{H_i,\beta_i}^*$ send $\beta(u,v)$ to the other nodes in $V_{H_i,\beta_i}^*$, in one round, and then having each node of $V_{H_i,\beta_i}^*$ send to $v^*$ the routing table entry corresponding to its shortest path to $c(H_i)$ in $G_i$, also in one round (the nodes have already computed this information in \ref{alg:QDLSI}. Then $v^*$ notifies the nodes that are part of $\zeta$, which can then add the appropriate edge to $T_i$, needing yet another round, so that step \ref{itm:unpacking2} can be done in three rounds of classical communication only. Step \ref{itm:unpacking3} is handled similarly. For the outgoing edge, each node in  $V_{H_i,\beta_i}^*$ sends $W_{vu}^{G_i}$ to the other nodes in $ V_{H_i,\beta_i}^*$ so that the appropriate edge to add to $T_i$ can be determined (in case of a tie, the node with smaller ID can be the one to add the edge), so this can be done in one round. For step \ref{itm:unpacking4}, every node in $\zeta$ notifies its neighbors that it is in $\zeta$, after which every node can determine which edges to add to $T_i$. For the unpacking of $V_{H_i,\beta_i}^*$, the information and communication for implementing its unpacking is contained in the nodes of $V_{H_i,\beta_i}^*$, so we can indeed unpack all vertices synchronously to obtain $G_{i}$ even when multiple super-vertices were contracted to get $G_{i+1}$. Hence, one layer of unpacking using this procedure can be implemented in $5$ rounds (making use of the APSP and routing table information computed earlier before the contractions in \ref{alg:QDLSI}). Since there are at most $\ceil{\log n}$ contraction steps, the unpacking procedure can be implemented in $5\cdot \ceil{\log n}$ rounds. 

\subsection{Information access}\label{app:information access}

In remark \ref{remark: information storage}, we mention that it suffices for all information regarding the input graph to be stored classically, with quantum access to it.
Here, we expand on what we mean by that and refer the interested reader to \cite{Booth2021_QuantumCP} for further details.

While our algorithms use quantum subroutines, the problem instances and their solutions are encoded as classical information.
The required quantum access refers to the ability to access the classical data so that computation in superposition of this data is possible.
For instance, in the standard (non-distributed) Grover search algorithm, with a problem instance described by a function $g: X \rightarrow {0, 1}$, we need the ability to apply the unitary $U_w |x\rangle = (-1)^{g(x)}|x\rangle$ to an $N$-qubit superposition state $|s\rangle = \frac{1}{\sqrt{N}}\sum_{x=0}^{N-1}|x\rangle$.
This unitary is also referred to as the ``oracle'', and a call to it as a ``query''.
If we wish to use the distributed Grover search in example \ref{ex: grover example}, in which the node $u$ leading the search tries to determine whether each edge $uv$ incident on it is part of a triangle in graph $G$, the unitary that node $v$ must be able to evaluate is the indicator function of its neighborhood, and $u$ must be able to apply the Grover diffusion unitary restricted to its neighborhood.
Then after initializing the $N$-qubit equal superposition, nodes $u$ and $v$ can send a register of qubits back and forth between each other, with $v$ evaluating the unitary corresponding to the indicator of its neighborhood and $u$ applying the Grover diffusion operator restricted to its neighborhood.%, on the superposition state $|\psi\rangle_0:= \frac{1}{\sqrt{|\mathcal{N}_G(u)|}}\sum_{x\in\mathcal{N}_G(u)}|x\rangle$. 
The same ideas transfer over to a distributed quantum implementation of the \ref{alg:EvaluationA} (or \ref{alg:EvaluationB}) procedure.
There, instead of evaluating unitaries corresponding to indicators, in step \ref{itm:eval2}, each node $v_{(i,j,k)}$ evaluates the unitary corresponding to the truth values of inequality \ref{eq:feval} for the evaluation steps.
That information is then returned to the node that sent it, which can then apply the appropriate Grover diffusion operator. 

In general, quantum random access memory (QRAM) is the data structure that allows queries to the oracle. We can use circuit QRAM in our protocols or could make use of special-purpose hardware QRAM if it were to be realized. This choice does not affect the number of rounds of communication but would affect the efficiency of computation at each node.
A main component of the distributed algorithms discussed in this work is quantum query access for each node to its list of edges and their weights in some graph $G$.
This information is stored in memory, and the QRAM implementing the query to retrieve it can be called in time $\mathcal{O}(\log n)$, resulting in a limited overhead for our algorithms.
This retrieval of information takes place {\em locally} at each node; hence, this overhead does not add to the round complexity of our algorithms in the CONGEST-CLIQUE setting. We refer to \cite{Giovannetti_2008_QRAM} for more details on QRAM.
\end{document}